\documentclass{IEEEtran}

\usepackage{cite}
\usepackage{amsthm, amsmath,amssymb,amsfonts,bbm}
\usepackage{algorithmic}
\usepackage[dvips]{graphicx}
\usepackage{subfigure}
\graphicspath{{graphics/}, {graphics/Results}, {graphics/Results/0_MatlabResults}}
\usepackage{makecell}
\usepackage{multirow}
\usepackage{filecontents, pgfplots}
\usepackage[algoruled]{algorithm2e}
\usepackage{tikz}
\usepackage{textcomp}
\usepackage{xcolor}
\usepackage{verbatim}
\usepackage[binary-units=true]{siunitx}
\DeclareSIUnit{\belmilliwatt}{Bm}
\DeclareSIUnit{\dBm}{\deci\belmilliwatt}

\usepackage[acronym,nomain]{glossaries}
\makeglossaries
\newacronym{swipt}{SWIPT}{simultaneous wireless information and power transfer}
\newacronym{slipt}{SLIPT}{simultaneous lightwave information and power transfer}
\newacronym{wpt}{WPT}{wireless power transfer}
\newacronym{wpcn}{WPCN}{wireless powered communication network}
\glsunset{wpcn}
\newacronym{wit}{WIT}{wireless information transfer}
\newacronym{awgn}{AWGN}{additive white Gaussian noise}
\newacronym{tx}{TX}{transmitter}
\newacronym{rx}{RX}{receiver}
\newacronym{bs}{BS}{base station}
\newacronym{ir}{IR}{information receiver}
\newacronym{eh}{EH}{energy harvesting}
\newacronym{id}{ID}{information detection}
\newacronym{irs}{IRS}{intelligent reflected surface}
\newacronym{ap}{AP}{average power}
\newacronym{pp}{PP}{peak power}
\newacronym{aoa}{AOA}{angle-of-arrival}
\newacronym{aod}{AOD}{angle-of-departure}
\newacronym{eec}{EEC}{electical equivalent circuit}

\newacronym{fso}{FSO}{free space optical}
\glsunset{fso}
\newacronym{vlc}{VLC}{visible light communication}
\newacronym{led}{LED}{light emitting diode}
\glsunset{led}
\newacronym{ook}{OOK}{On-Off keying}
\newacronym{pam}{PAM}{pulse amplitude modulated}
\newacronym{tdma}{TDMA}{time division multiple access}
\newacronym{ml}{ML}{maximum likelihood}

\newacronym{siso}{SISO}{single-input single-output}
\newacronym{mimo}{MIMO}{multiple-input multiple-output}
\newacronym{miso}{MISO}{multiple-input single-output}
\newacronym{simo}{SIMO}{single-input multiple-output}

\newacronym{rf}{RF}{radio frequency}
\newacronym{dc}{DC}{direct current}
\newacronym{ac}{AC}{alternating current}
\newacronym{papr}{PAPR}{peak-to-average power ratio}
\newacronym{lp}{LPF}{low-pass filter}
\newacronym{mc}{MC}{matching circuit}
\newacronym{mrt}{MRT}{maximum ratio transmission}
\newacronym{ecb}{ECB}{equivalent complex baseband}
\newacronym{tdd}{TDD}{time-division-duplex}
\newacronym{zf}{ZF}{zero forcing}
\newacronym{snr}{SNR}{signal-to-noise ratio}
\newacronym{sinr}{SINR}{signal-to-interference-plus-noise ratio}

\newacronym{rv}{RV}{random variable}
\newacronym{iid}{i.i.d.}{independent and identically distributed}
\newacronym{Pdf}{Pdf}{Probability density function}
\newacronym{pdf}{pdf}{probability density function}
\newacronym{cdf}{cdf}{cumulative distribution function}

\newacronym{dnn}{DNN}{dense neural networks}
\glsunset{dnn}
\newacronym{mdp}{MDP}{Markov decision process}

\newacronym{sca}{SCA}{successive convex approximation}
\newacronym{sdr}{SDR}{semi-definite relaxation}

\newacronym{spr}{LP}{low power}
\newacronym{mpr}{MP}{medium power}
\newacronym{lpr}{HP}{high power}

\begin{document}

	\newtheorem{proposition}{Proposition}	
	\newtheorem{lemma}{Lemma}	
	\newtheorem{corollary}{Corollary}
	\newtheorem{assumption}{Assumption}	
	\newtheorem{remark}{Remark}	
	
	\title{Accurate EH Modelling and Achievable Information Rate for SLIPT Systems with Multi-Junction Photovoltaic Receivers}

	\author{Nikita Shanin,
			Hedieh Ajam,
			Vasilis K. Papanikolaou,
			Laura Cottatellucci,
			and Robert Schober}

	\maketitle

\begin{abstract}
	\let\thefootnote\relax\footnotetext{This paper was accepted in part for the presentation at the IEEE Global Communications Conference, 2023 \cite{Shanin2023c}.

This work was supported in part by the German Science Foundation through project SFB 1483 - Project-ID 442419336, EmpkinS.

Nikita Shanin, Hedieh Ajam, Vasilis K. Papanikolaou, Laura Cottatellucci, and Robert Schober are with the Institute for Digital Communications, Friedrich-Alexander-Universit{\"a}t (FAU) Erlangen-N{\"u}rnberg, 91058 Erlangen, Germany (e-mail: nikita.shanin@fau.de; hedieh.ajam@fau.de; vasilis.papanikolaou@fau.de; laura.cottatellucci@fau.de; robert.schober@fau.de).
}
In this paper, we study optical simultaneous lightwave information and power transfer (SLIPT) systems employing photovoltaic optical receivers (RXs).
To be able to efficiently harvest energy from broadband light, we propose to employ multi-junction photovoltaic cells at the optical RX.
We consider the case, where the optical RX is illuminated by ambient light, an intensity-modulated information free space optical (FSO) signal, and since the ambient light may not be always present, e.g., indoor or at night, a dedicated energy-providing broadband optical signal.
Based on the analysis of the equivalent electrical circuit of the multi-junction photovoltaic RX, we model the current flow through the photovoltaic cell and derive a novel accurate and two novel approximate energy harvesting (EH) models for the two cases, where the optical RX is equipped with a single and multiple p-n junctions, respectively. 
Taking into account the non-linear behavior of the photovoltaic RXs, we derive the optimal distribution of the transmit information signal that maximizes the achievable information rate.
Furthermore, for a practical On-Off keying (OOK) modulated FSO information signal, we obtain the bit-error rate at the RX.
We validate the proposed EH models by circuit simulations and show that the photovoltaic RXs saturate for high received signal powers.
For single-junction RXs, we compare the proposed EH model with two well-known baseline models and demonstrate that, in contrast to the proposed EH model, they are not able to fully capture the RX non-linearity. 
Moreover, since multi-junction RXs allow a more efficient allocation of the optical power, they are more robust against saturation, and thus, are able to harvest significantly more power and achieve higher data rates than RXs employing a single p-n junction
Finally, we highlight a tradeoff between the information rate and harvested power in SLIPT systems and demonstrate that the proposed optimal distribution yields significantly higher achievable information rates compared to uniformly distributed transmit signals, which are optimal for linear optical information RXs.


\end{abstract}
\section{Introduction}
\label{Section:Introduction}
Free space optical (FSO) communication systems support extremely high data rates by focusing an intensity-modulated laser signal on a small optical \gls*{rx} \cite{Zhang2019, Ajam2020}.
Since practical RXs of FSO systems convert the received optical intensity-modulated signal into an electrical amplitude-modulated signal employing light-sensitive devices, e.g., photodetectors and photovoltaic cells, they are also able to harvest energy from the received signal \cite{Luque2010, Mertens2014}.
This aspect has recently fueled interest in \gls*{slipt} systems, where the laser signal is exploited not only to convey information, but also to deliver power to user devices \cite{MaShuai2019, Li2017, Pan2019, Diamantoulakis2018, Tran2019, Sepehrvand2021, Papanikolaou2021, Zhang2018, Makki2018, Wang2015, Fakidis2020}.

SLIPT systems employing RXs equipped with photoelectric circuits for \gls*{eh} and information decoding were considered in \cite{MaShuai2019, Li2017, Pan2019, Diamantoulakis2018, Tran2019, Sepehrvand2021}.
In these systems, the \gls*{tx} was equipped with light emitting diodes (LEDs) to cover a large area, provide indoor illumination, and simultaneously deliver data and energy to user devices.
Since photovoltaic cells and photodiodes are preferable for EH and information reception, respectively, the authors in \cite{MaShuai2019} proposed to deploy both types of optoelectronic devices at the \gls*{rx} for SLIPT.
However, since the surface area of practical user equipments may be constrained to be small, it may not be possible to mount multiple optoelectronic devices on the same RX.
Therefore, in \cite{Li2017, Pan2019, Diamantoulakis2018, Tran2019, Sepehrvand2021}, the authors studied SLIPT systems, where the RXs were equipped with single photovoltaic cells.
In particular, the SLIPT system in \cite{Li2017} comprised multiple users that were served via a \gls*{tdma} protocol, and in each time slot, only one RX decoded information while the other RXs harvested power from the received optical signal.
However, transmission protocols based on time-sharing between information and power transfer are not optimal for SLIPT \cite{Pan2019}.
Therefore, the user equipments in \cite{Diamantoulakis2018, Tran2019, Sepehrvand2021} were designed to receive both power and information in the same time slot.
To this end, the authors in \cite{Diamantoulakis2018} proposed to separate the \gls*{dc} and \gls*{ac} components of the received signal and utilize them for EH and information decoding at the RX, respectively.
Employing the RX design in \cite{Diamantoulakis2018}, the authors of \cite{Tran2019} and  \cite{Sepehrvand2021} studied resource allocation for SLIPT and hybrid radio frequency (RF)-optical simultaneous wireless information and power transfer (SWIPT) systems, respectively.
Their analysis highlighted a tradeoff between the achievable information rate and the harvested power at the RXs, which was characterized by a rate-power region.

Since the transmit power of optical light sources is typically limited due to hardware constraints and eye-safety regulations \cite{Lasersafety2014}, the power received at a small photovoltaic RX is typically rather low if undirected LEDs are employed at the TX.
Therefore, the authors in \cite{Makki2018, Papanikolaou2021, Zhang2018, Wang2015, Fakidis2020} proposed to exploit a directed FSO laser beam to transfer more power to user devices.
In particular, the authors of \cite{Papanikolaou2021} and \cite{Zhang2018} analyzed the impact of imperfect channel knowledge at the TX on the efficiency of FSO wireless power transfer.
Furthermore, for FSO SLIPT systems, the authors in \cite{Makki2018} optimized the resource allocation policy and determined the total information throughput.

For the SLIPT systems in \cite{MaShuai2019, Li2017, Diamantoulakis2018, Pan2019, Tran2019, Sepehrvand2021, Papanikolaou2021, Makki2018, Zhang2018}, the EH was designed by tracking the maximum power point (MPP) at the optical photovoltaic cell.
In other words, the load parameters at the RXs in \cite{MaShuai2019, Li2017, Diamantoulakis2018, Pan2019, Tran2019, Sepehrvand2021, Papanikolaou2021, Makki2018, Zhang2018} were tuned to maximize the EH efficiency for a given received signal power, which, in general, requires RX loads with variable impedances, and thus, is not feasible for the design of a practical RX, where the electrical properties of the load device are typically fixed \cite{Horowitz1989}.
To avoid this impractical assumption, the authors in \cite{Wang2015} considered optical RXs with fixed load resistances, derived a non-linear EH model characterizing the harvested power as a function of the received power, and studied the frequency response of the optical RX for the received information signal and the noise impairing the information decoding.
Furthermore, based on the non-linear EH model and the RX design in \cite{Wang2015}, the authors in \cite{Fakidis2020} implemented an FSO SLIPT system based on a GaAs photovoltaic cell.
Their analysis confirmed that the non-linearities of photovoltaic cells have to be carefully taken into account for efficient SLIPT system design.

Although the SLIPT system designs in \cite{Wang2015} and \cite{Fakidis2020} were experimentally shown to achieve high data rates, they still relied on several assumptions that limit SLIPT system performance.
First, in \cite{MaShuai2019, Li2017, Diamantoulakis2018, Pan2019, Tran2019, Sepehrvand2021, Papanikolaou2021, Makki2018, Zhang2018, Wang2015, Fakidis2020}, the authors modelled the photovoltaic cell with an \gls*{eec} employing a single non-linear diode.
However, for high power received signals, the electrical current at the output of the photovoltaic cell is mainly determined by the diffusion of electrons and holes in the p-n junction of the cell, while for low power received signals, diffusion is negligible and the photovoltaic current is caused by particle recombination in the depletion region of the junction \cite{Luque2010, Mertens2014}.
Thus, an accurate \gls*{eec} of a photovoltaic RX has to comprise two diodes to model the non-linearities of the output current in both the low and high input power regimes \cite{Mertens2014, Luque2010}.
Furthermore, since practical photovoltaic RXs exhibit a non-linear behavior, the received ambient light has an impact on the performance of EH and information decoding at the RX, whereas it was neglected in the analysis of SLIPT systems in \cite{MaShuai2019, Li2017, Diamantoulakis2018, Pan2019, Tran2019, Sepehrvand2021, Papanikolaou2021, Makki2018, Zhang2018, Wang2015, Fakidis2020}.
Moreover, while the optical RXs in \cite{MaShuai2019, Li2017, Diamantoulakis2018, Pan2019, Tran2019, Sepehrvand2021, Papanikolaou2021, Makki2018, Zhang2018, Wang2015, Fakidis2020} employed a single p-n junction, recent advances in the design of photovoltaic cells have provided significant improvements in EH efficiency by employing multiple semiconducting elements, and thus, multiple p-n junctions \cite{Dimroth2016, Geisz2020}.
In fact, broadband multi-junction photovoltaic cells yield significantly higher EH efficiencies than single-junction cells when illuminated by a broadband optical signal, e.g., sunlight or ambient indoor illumination.
However, ambient optical light is intermittent in nature and not always available, e.g., at night \cite{Luque2010, Mertens2014}.
Furthermore, we note that for the RX design in \cite{MaShuai2019, Li2017, Diamantoulakis2018, Pan2019, Tran2019, Sepehrvand2021, Papanikolaou2021, Makki2018, Zhang2018, Wang2015, Fakidis2020}, the authors assumed a linear frequency response of the RX circuit with respect to the received information signal, which may not be realistic due to the non-linear behavior of photovoltaic cells \cite{Luque2010, Mertens2014}.
Finally, for practical photovoltaic RXs, optimal filtering of the information signal is not feasible due to the RX non-linearities and the extremely high sampling rates required for accurate discretization of the output information signal \cite{Tietze2012, Horowitz1989}.
In \cite{Shanin2023c}, which is the conference version of this paper, we derived the distribution of the transmitted information-carrying signal that maximizes the achievable information rate in FSO SLIPT systems employing single-junction photovoltaic RXs.
In this paper, taking into account all non-linearities of optical photovoltaic RXs equipped with multiple p-n junctions, we derive an accurate EH model, the optimal input distribution for information transmission, the maximum achievable information rate, and the bit-error rate of SLIPT systems.

In this work, we consider a \gls*{slipt} system, where the optical RX is equipped with a photovoltaic cell.
Since photovoltaic RXs with multiple p-n junctions are efficient for EH from broadband optical signals, we propose to employ a multi-junction photovoltaic cell at the RX. 
We consider the case where the RX is illuminated by ambient light, an intensity-modulated information FSO signal, and since the ambient light may not be present, e.g., indoor or at night, a dedicated energy-providing broadband optical signal.
The optical signal received by the photovoltaic RX is converted to an electrical signal, whose AC and DC components are separated and utilized for information reception and EH, respectively.
The main contributions of this work can be summarized as follows:
\begin{itemize}
	\item Taking into account the non-linearities of photovoltaic RXs with multiple p-n junctions, we thoroughly model the current flow through the RX and derive a novel accurate EH model to characterize the instantaneous harvested power.
	Furthermore, since the accurate EH model may not lend itself to SLIPT system design, we also propose two novel approximate analytical closed-form EH models for photovoltaic RXs equipped with one and multiple p-n junctions, respectively. 
	\item  Since optimal filtering of the electrical information signal is not possible due to the non-linearity of photovoltaic RXs and the high required sampling rates, we design a suboptimal information decoder. 
	Furthermore, taking into account the non-linear behavior of the optical RX, we derive the optimal distribution of the information-carrying transmit signal that maximizes the achievable information rate for SLIPT systems.
	Moreover, for practical \gls*{ook} modulated information signals, we also derive the bit-error rate for maximum likelihood information detection at the photovoltaic RX.
	\item We validate the proposed EH models by circuit simulations and demonstrate that for single-junction RXs, two baseline EH models based on MPP tracking at the RX and a single-diode p-n junction, respectively, are not able to accurately capture the non-linearity of photovoltaic RXs.
	We show that the proposed optimal distribution yields significantly higher achievable information rates compared to uniformly distributed transmit signals, which are optimal for linear optical information RXs.
	Furthermore, we demonstrate that photovoltaic RXs saturate when the received optical power is high and multi-junction cells are able to achieve significantly higher harvested powers and data rates than single-junction RXs.
	Finally, we study the tradeoff between the achievable data rate and the average harvested power at the RX and show that by adjusting the power of the energy-providing optical signal, the achievable data rate of SLIPT systems can be traded for a higher harvested power.
\end{itemize}

The remainder of this paper is organized as follows.
In Section \ref{Section:SysModel}, we discuss the considered system model.
Next, in Section \ref{Section:EH}, we present the \gls*{eec} of the RX and model the EH at multi-junction photovoltaic RXs.
Then, in Section \ref{Section:ID}, we investigate the information reception at the considered optical RXs.
In Section~\ref{Section:NumResults}, we provide numerical results.
Finally, we draw conclusions in Section~\ref{Section:Conclusions}.

{\itshape Notations:} Bold lower case letters $\boldsymbol{x}$ represent vectors and $x_i$ is the $i^\text{th}$ element of $\boldsymbol{x}$.
The sets of real, real non-negative, and natural numbers are represented by $\mathbb{R}$, $\mathbb{R}_{+}$, and $\mathbb{N}$, respectively.
$f(x;y)$ stands for a function of variable $x$ parametrized by $y$.
The inverse, first-order derivative, and domain of the one-dimensional function $f(\cdot)$ are denoted by $f^{-1}(\cdot)$, $f'(\cdot)$, and $\text{dom} \{f\}$, respectively.
Furthermore, $f_s(\cdot)$ and $F_s(\cdot)$ denote the \gls*{pdf} and \gls*{cdf} of random variable $s$, respectively.
$\delta(\cdot)$ is the Dirac delta function.
Finally, $\Pr\{A\}$ and $\mathbb{E}\{s\}$ stand for the probability of event $A$ and statistical expectation of random variable $s$, respectively.

\section{System Model}
\label{Section:SysModel}

\begin{figure}[!t]
	\centering
	\includegraphics[draft=false, width=1\linewidth]{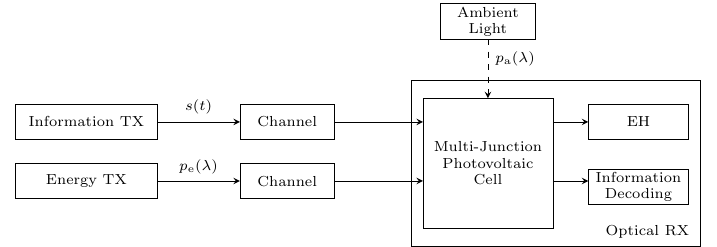}
	\caption{SLIPT system model comprising an information FSO TX, an optical energy TX, an ambient light source, and a broadband multi-junction photovoltaic RX \cite{Wang2015, Fakidis2020, Dimroth2016, Luque2010}.}
	\label{Fig:SystemModel}\vspace*{-5pt}
\end{figure}

We consider a \gls*{slipt} system, where an information TX sends an intensity-modulated FSO\footnotemark\hspace*{0pt} signal focused on an optical RX equipped with a broadband photovoltaic cell comprising $N \in \mathbb{N}$ p-n junctions \cite{Mertens2014, Luque2010, Geisz2020, Dimroth2016}.
\footnotetext{In this paper, as in \cite{Wang2015, Fakidis2020}, to achieve higher data rates, we assume that the optical RX is illuminated with an information-bearing FSO signal. We note that the extension of our analysis to SLIPT systems employing LEDs for information transmission is possible.}
Multi-junction photovoltaic cells are able to efficiently harvest power from broadband optical signals \cite{Mertens2014, Luque2010}.
Here, we assume that the optical RX is illuminated by ambient light, e.g., sunlight or indoor illumination.
Furthermore, since the EH efficiency of multi-junction photovoltaic RXs may be compromised when the received ambient light intensity is low, e.g., at night \cite{Luque2010, Mertens2014}, we additionally assume that the photovoltaic RX is illuminated by a dedicated energy-providing broadband optical signal, as shown in Fig.~\ref{Fig:SystemModel}.

We express the power spectral density of the stochastic process modelling the received light as a function of wavelength $\lambda$ and time $t, t \in [kT, (k+1)T]$, as follows:
\begin{align}
	p_\text{r}(\lambda, t) = h s(t) \delta(\lambda - \lambda_0)  + g(&\lambda) p_\text{e}(\lambda) \nonumber \\
	&+ p_\text{a}(\lambda) + \tilde{w}(\lambda, t).
	\label{Eqn:ReceivedSpectrum}
\end{align} 
Here, $s(t) = s[k] \psi(t-kT)$ is the power of the intensity-modulated information FSO signal with carrier wavelength $\lambda_0$, where $s[k], \forall k \in \mathbb{N},$ is the information symbol transmitted in time slot $k, k\in \mathbb{N},$ and modelled as an \gls*{iid} realization of a scalar non-negative random variable $s$, and $\psi(t)$ is a rectangular pulse of duration $T$, i.e., $\psi(t)$ takes value 1 if $t\in[0, T)$ and $0$, otherwise.
Since the intensity of FSO signals is limited by hardware constraints and eye safety regulations \cite{Lasersafety2014}, we limit the transmit power of the information signal to $A^2$, and thus, the pdf of $s$ satisfies $\text{dom} \{f_{s}\} \subseteq [0, A^2]$.
In (\ref{Eqn:ReceivedSpectrum}), $p_\text{e}(\lambda)$ and $g(\lambda)$ are the static power spectral density of the optical energy-providing signal and the frequency response of the channel between the energy TX and the optical RX, respectively.
Furthermore, the scalar channel gain between the information FSO TX and the optical RX and the power spectral density of the ambient light at the RX are denoted by $h \in \mathbb{R}_{+}$ and $p_\text{a}(\lambda)$, respectively.
We assume that the power spectral densities $p_\text{e}(\lambda)$, $p_\text{a}(\lambda)$ and the channels $h$, $g(\lambda)$ are perfectly known\footnotemark\hspace*{0pt} at the information FSO TX and the optical RX.
\footnotetext{In this work, we assume that the power spectral density of the energy-providing signal is fixed and known at the information TX and the optical RX. We note that the derivation of the optimal energy signal for SLIPT is an interesting direction for future work and is not tackled in this paper.}
Finally, $\tilde{w}(\lambda, t)$ is the time-varying power spectral density of the noise received at the photovoltaic RX.
The input noise at the RX is a non-stationary stochastic process caused by the random fluctuations of the ambient light, the intensity of the optical energy-providing signal, and the laser source \cite{Lapidoth2009}.
Moreover, the output signal at the photovoltaic RX is additionally impaired by thermal and shot noise, which are generated by the resistances and p-n junctions of the RX, respectively \cite{Lapidoth2009, Wang2015}.
For the derivation of the EH model in Section~\ref{Section:EH}, we neglect the noise at the RX since its contribution to the average harvested power is negligible \cite{MaShuai2019, Tran2019, Sepehrvand2021, Wang2015}.
To analyze the performance of information transmission, in Section~\ref{Section:ID}, we assume that the thermal noise at the RX dominates and we model the equivalent noise, which impairs the output information symbols, as \gls*{awgn} \cite{Wang2015, Lapidoth2009}.

\section{EH at Multi-Junction Photovoltaic RXs}
\label{Section:EH}
In this section, we study the \gls*{eh} at the multi-junction photovoltaic RX.
To this end, we first present the \gls*{eec} of the RX, where the AC and DC components of the electrical signal at the output of the photovoltaic cell are separated and utilized for information reception and EH, respectively.
Next, we propose a novel accurate model for characterizing the EH at the user device.
Finally, since the derived accurate EH model is not practical for the design of SLIPT systems, we propose two novel approximate analytical closed-form EH models for single-junction and multi-junction photovoltaic RXs, respectively.
\subsection{EEC of Multi-Junction Photovoltaic RXs}
	\label{Section:EEC}

\begin{figure}
	\centering
		\includegraphics[draft=false, width=1\linewidth]{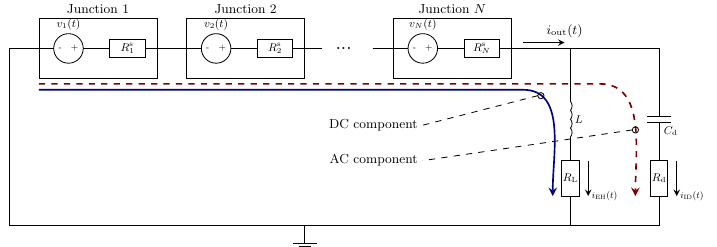}
		\caption{EEC of multi-junction photovoltaic RXs, where the DC (blue solid line) and AC (red dashed line) components of the electrical output signal $i_\text{out}(t)$ are separated and utilized for \gls*{eh} and information decoding, respectively \cite{Mertens2014, Luque2010, Horowitz1989}.}
		\label{Fig:MjCell}\vspace*{-10pt}
\end{figure}

\begin{figure}
	\centering
	\includegraphics[draft=false, width=1\linewidth]{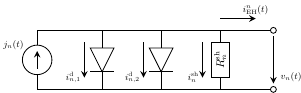}
	\caption{EEC of p-n junction $n, n\in\{1,2,\cdots, N\},$ of the multi-junction photovoltaic RX \cite{Mertens2014, Wang2015, Fakidis2020}.}
	\label{Fig:SingleJunction}\vspace*{-10pt}
\end{figure}

In the following, we present the EEC of the multi-junction photovoltaic cell employed at the optical RX.
The EEC is shown in Fig.~\ref{Fig:MjCell} and employs a series of $N$ voltage sources, where each source represents the output voltage of a single p-n junction.
In particular, the voltage $v_n(t)$ generated by source $n, n \in \{1,2,\cdots, N\},$ in Fig.~\ref{Fig:MjCell} is the voltage at the output of p-n junction $n$, whose EEC is shown in Fig.~\ref{Fig:SingleJunction}, and $R^\text{s}_n, n\in\{1,2,\cdots, N\},$ is the corresponding equivalent series resistance at the output of p-n junction $n$ \cite{Wang2015, Mertens2014}.

Each p-n junction of a multi-junction photovoltaic cell can be equivalently modelled by an EEC comprising two diodes that represent the diffusion current in the neutral region and the particle recombination in the depletion region of the junction, respectively, see Fig.~\ref{Fig:SingleJunction} \cite{Mertens2014, Luque2010}.
In particular, at low received light intensities, the diffusion current is low compared to the particle recombination in the depletion region and the corresponding recombination current can be modelled by a diode with ideality factor $2$ \cite{Mertens2014}.
However, if the received light intensity is high, the diffusion current dominates and can be modelled by a diode with ideality factor $1$ \cite{Luque2010}.
Furthermore, we model the parasitic shunt resistance of p-n junction $n, n\in\{1,2,\cdots, N\}$, by resistor $R^\text{sh}_n$ in Fig.~\ref{Fig:SingleJunction} \cite{Luque2010, Mertens2014}.

To efficiently harvest power and decode information from the received optical signal, similar to \cite{Diamantoulakis2018, Tran2019, Sepehrvand2021, Wang2015, Fakidis2020}, we split the AC and DC signal components via an RL low-pass filter and an RC high-pass filter and utilize these signal components for information reception and EH, respectively, as shown in Fig.~\ref{Fig:MjCell} \cite{Horowitz1989}.
In particular, the energy provided by the DC signal component is harvested at a resistive EH load $R_\text{L}$, whereas the AC component is received at the information RX modelled by resistance $R_\text{d}$.
As in \cite{Diamantoulakis2018, Tran2019, Sepehrvand2021, Wang2015, Fakidis2020}, we assume ideal lossless high- and low-pass electrical filters, i.e., we neglect the parasitic resistances of inductance $L$ and capacitance $C_\text{d}$ and suppose that the DC and AC currents flow only through resistors $R_\text{L}$ and $R_\text{d}$, respectively.

We note that the time interval $T$, capacitance $C_\text{d}$ and inductance $L$ shown in Fig.~\ref{Fig:MjCell} can always be tuned to avoid undesired memory effects due to the charging and discharging of reactive elements, i.e., capacitances and inductances, in the photovoltaic RX \cite{Horowitz1989}.
Thus, after a transient phase that we assume to be negligible, in each time slot $k, \forall k,$ the photovoltaic RX reaches a steady state and we can neglect the dependence of the output currents $i_\text{out}(t)$, $i_\text{EH}(t)$, and $i_\text{ID}(t)$ in time slot $k, \forall k$, see Fig.~\ref{Fig:MjCell}, on information symbols transmitted prior to that time slot, i.e., symbols $s[m]$, $m < k$, $m \in \mathbb{N}$ \cite{Tran2019, Wang2015}, as shown in Fig.~\ref{Fig:InpulseResponse}.

We model the wavelength dependent spectral response $r_n(\lambda)$ of p-n junction $n$ of the multi-junction photovoltaic cell as follows \cite{Mertens2014, Luque2010}:
\begin{equation}
	r_n(\lambda) = \lambda \mu_n(\lambda), \forall n \in \{1,2,\cdots, N\},
\end{equation}
\noindent where $\mu_n(\lambda)$ is the conversion efficiency of p-n junction $n$ as a function of wavelength $\lambda$.
In practical systems, $\mu_n(\lambda)$ is typically almost constant within the passband of p-n junction $n, \forall n$ \cite{Mertens2014, Luque2010}. 
Next, we express the photovoltaic current $j_n(t)$ induced in p-n junction $n, \forall n,$ of the multi-junction photovoltaic cell as follows \cite{Mertens2014}:
\begin{align}
	j_n(t) &= \int p_\text{r}(\lambda, t) r_n(\lambda) d\lambda =  {h} s(t) r_n(\lambda_0) + w_n(t) \nonumber \\ &
	\;+  \int p_\text{e}(\lambda) g(\lambda) r_n(\lambda) d\lambda + \int p_\text{a}(\lambda) r_n(\lambda) d\lambda  ,
	\label{Eqn:InducedCurrent}
\end{align}
\noindent where $w_n(t) = {\int \tilde{w}(\lambda, t) r_n(\lambda)  d\lambda }$ is the photovoltaic current in p-n junction $n, \forall n,$ due to the RX noise.

\subsection{Derivation of EH Models}
\label{Section:EH_Model}

\begin{figure}[!t]
	\centering
	\includegraphics[draft=false, width=1\linewidth]{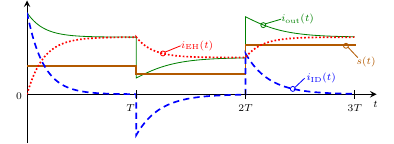} 
	\caption{Information-carrying transmit signal $s(t)$ and currents $i_\text{out}(t)$ (solid green line), $i_\text{ID}(t)$ (dashed blue line), and $i_\text{EH}(t)$ (dotted red line) obtained for the EEC in Fig.~\ref{Fig:MjCell} with $N = 1$ and the circuit parameters specified in Table~\ref{Table:SimulationSetupEEC}.}
	\label{Fig:InpulseResponse}\vspace*{-10pt}
\end{figure}

In the following, we derive an accurate and two approximate analytical closed-form EH models for the considered photovoltaic RX.
First, as in \cite{MaShuai2019, Li2017, Diamantoulakis2018, Pan2019, Tran2019, Sepehrvand2021, Papanikolaou2021, Zhang2018, Makki2018, Wang2015}, and other related works, we neglect the impact of noise on the power harvested at the RX.
Next, since an ideal low-pass RL filter is assumed for EH, we neglect the ripples of the voltage across the EH load resistance $R_\text{L}$ during the reception of information symbol $s[k], \forall k$ \cite{Tietze2012}.
Thus, in steady state, we have $i_\text{ID}(t) \approx 0$ and the corresponding equivalent current induced in p-n junction $n, n\in\{1,2,\cdots, N\},$ in time slot $k, \forall k$, is given by
\begin{align}
	j_n[k] \triangleq j(kT) = {h} s[k] r_n(\lambda_0) + j^\text{a}_n 
	\label{Eqn:EquivalentInducedCurrent}
\end{align}
\noindent where $j^\text{a}_n = \int p_\text{e}(\lambda) g(\lambda) r_n(\lambda) d\lambda + \int p_\text{a}(\lambda) r_n(\lambda) d\lambda$ is the photovoltaic current in p-n junction $n, n\in\{1,2,\cdots, N\},$ due to the optical energy signal and ambient light.

We note that the p-n junctions of a broadband multi-junction photovoltaic cell are typically arranged to operate as optical wavelength splitters, i.e., the passbands of any two junctions do not intersect \cite{Luque2010, Dimroth2016}.
Thus, without loss of generality, we assume that only the first p-n junction at the RX is responsive to the FSO signal at wavelength $\lambda_0$, i.e., we have $r_n(\lambda_0) = 0, \forall n \in\{2,3,\cdots, N\},$ and $r_1(\lambda_0) \geq 0$.
Thus, the equivalent photovoltaic currents in (\ref{Eqn:EquivalentInducedCurrent}) are given by
\begin{align}
	j_1[k] &= j^\text{a}_1 + {h} s[k] r_1(\lambda_0), \label{Eqn:IndCurrent1} \\
	j_n[k] &= j^\text{a}_n, \; \forall n \in\{2,3,\cdots, N\}. \label{Eqn:IndCurrentN}
\end{align}
In the following, for a given vector of induced currents $\boldsymbol{j}^\text{a} = [j^\text{a}_1, j^\text{a}_2, \cdots, j^\text{a}_N] \in \mathbb{R}_{+}^N$ and a given current $j^\text{s}[k] = {h} s[k] r_1(\lambda_0)$, we characterize the corresponding instantaneous harvested power $P_\text{harv}$ at the EH load resistance $R_\text{L}$.

\begin{table}[!t]
	\centering
	\caption{Parameters of the RX circuit in Figs.~\ref{Fig:MjCell} and \ref{Fig:SingleJunction} \cite{Wang2015}.}
	\begin{tabular}{| p{2cm} | p{2cm} |}
		\hline
		$C_\text{d} = \SI{2.5}{\micro\farad}$ &$R_\text{d} = \SI{10}{\kilo\ohm}$,\\ 
		$L = \SI{10}{\milli\henry}$ &$R_\text{L} = \SI{10}{\kilo\ohm}$, \\ 
		\hline
		\multicolumn{2}{|c|}{ $I_{n} = I_{n,1} = I_{n,2} = \SI{1}{\nano\ampere}, \forall n$ }  \\
		\multicolumn{2}{|c|}{ $R^\text{sh}_n = \SI{100}{\mega\ohm}$, $R^\text{s}_n = \SI{100}{\ohm}, \forall n$ }  
		\\
		\hline
	\end{tabular}
	\label{Table:SimulationSetupEEC}\vspace*{-10pt}
\end{table}

	\subsubsection{Accurate EH Model}
	\label{Section:PreciseEH}
	
Let us express the output current of p-n junction $n, n\in\{1,2,\cdots, N\}$, whose EEC is shown in Fig.~\ref{Fig:SingleJunction}, as follows\footnotemark\hspace*{0pt} \cite{Mertens2014, Tietze2012}:
\footnotetext{For the simplicity of presentation, in the remainder of this section, we drop the time slot index $k$.}
\begin{align}
	i^n_{\text{EH}} &= j_n - i^\text{d}_{n,1} - i^\text{d}_{n,2}  - i^n_{\text{sh}} \nonumber\\ 
	&=  j_n - I_{n,1} (e^{\frac{ v_n }{V_T}} - 1)  \nonumber\\
	&\quad\quad\quad\quad- I_{n,2} (e^{\frac{ v_n }{2 V_T}} - 1)- \frac{ v_n }{R^\text{sh}_n} \triangleq \Phi_n( v_n ),
	\label{Eqn:OutputCurrent_Junction}
\end{align}
where $i^\text{d}_{n, 1}$ and $i^\text{d}_{n, 2}$ model the diffusion current and the particle recombination in the depletion region of p-n junction $n, \forall n$, respectively \cite{Luque2010}.
Here, $I_{n,m}, m\in\{1,2\}, n\in\{1,2,\cdots, N\}$, and $V_T$ are the reverse-bias saturation current of diode $m$ of p-n junction $n$ and the thermal voltage of the diode, respectively.

Since the p-n junctions of the multi-junction photovoltaic cell are connected in-series, as shown in Fig.~\ref{Fig:MjCell}, the output currents of the p-n junctions are identical, i.e., we have $i_{\text{EH}} = i^n_{\text{EH}}, n\in\{1,2,\cdots, N\}$ \cite{Tietze2012, Mertens2014}.
As we neglect the ripples of the voltage across the EH load resistance, the harvested power at the EH load can be obtained as 
\begin{equation}
	P_\text{harv} = R_\text{L} i^{*^2}_\text{EH},
	\label{Eqn:PharvCurrent}
\end{equation} 
\noindent where $i^{*}_\text{EH}$ is the solution of the following equation:
\begin{equation}
	i_\text{EH} =\frac{ \sum_n v_n }{R_\Sigma}= \frac{ \sum_n \Phi^{-1}_n (i_\text{EH}) }{R_\Sigma}
	\label{Eqn:OutputCurrent}
\end{equation}
\noindent with $R_\Sigma = \sum_n R_n^{\text{s}} + R_\text{L}$.
We note that since all functions $\Phi_n(\cdot), n\in\{1,2,\cdots, N\},$ in (\ref{Eqn:OutputCurrent}) are monotonically decreasing and non-negative in their domains, so are functions $\Phi^{-1}_n(\cdot), \forall n,$ and hence, their sum $\sum_n \Phi^{-1}_n(\cdot)$.
Therefore, $i_\text{EH}^*$ as solution of (\ref{Eqn:OutputCurrent}) can be interpreted as the intersection point of a monotonically decreasing non-negative function and the identity line.
Hence, (\ref{Eqn:OutputCurrent}) is a fixed point equation and the solution $i_\text{EH}^*$ of (\ref{Eqn:OutputCurrent}) is unique.

Thus, (\ref{Eqn:PharvCurrent}) and (\ref{Eqn:OutputCurrent}) accurately characterize the input-output behavior of a multi-junction photovoltaic RX.
Additionally, (\ref{Eqn:OutputCurrent}) can be solved numerically for any given $j^\text{s}$ and vector $\boldsymbol{j}^\text{a}$.
However, obtaining the value of functions $\Phi^{-1}_n (i_\text{EH}), \forall n,$ in closed form may be unfeasible.
Since complicated expressions hinder the ability to extract valuable insights into the behavior of photovoltaic cells, the derived accurate EH model may be impractical for SLIPT system design.
Therefore, in the following, we simplify (\ref{Eqn:OutputCurrent}) and derive an accurate and an approximate analytical closed-form EH model for $N = 1$ and $N \geq 1$, respectively.

	\subsubsection{Analytical Closed-form EH Models}
	\label{Section:ApproximateEH}
	
\begin{table*}[!t]
	\renewcommand{\arraystretch}{2}
	\centering
	\caption{Accurate and approximate EH models for single-junction and multi-junction photovoltaic RXs.}
	\begin{tabular}{| l | l |}
		\hline
		Accurate EH model & \makecell[l]{
			$P_\text{harv} = i^{*^2}_\text{EH} R_\text{L}$, where $i^*_\text{EH}$ is the solution of equation $i^*_\text{EH} R_\Sigma = \sum_n \Phi^{-1}_n (i_\text{EH}), $ \\ $\Phi_n(v_n) = j_n - I_{n,1} (e^{\frac{ v_n }{V_T}} - 1) - I_{n,2} (e^{\frac{ v_n }{2 V_T}} - 1) - \frac{ v_n }{R^{\text{sh}}_n}$, and $j_n, \forall n,$ is given in (\ref{Eqn:EquivalentInducedCurrent})
		}\\
		\hline
		Approximate EH model & \makecell[l]{
			$P_\text{harv} = i^{*^2}_\text{EH} R_\text{L}$, where $i^*_\text{EH}$ is the solution of equation \\ $\prod_{n=1}^N ( j_n - i_{\text{EH}} + I_n ) \exp \Big[ - \frac{R_\Sigma}{V_T} i_\text{EH} \Big] - \prod_{n=1}^N I_n = 0$ and $j_n, \forall n,$ is given in (\ref{Eqn:EquivalentInducedCurrent})
		}\\
		\hline
		Single-junction cells & $P_\text{harv}(j^\text{s}; j^\text{a}_1) = R_\text{L} \bigg\{j^\text{s} + j^\text{a}_1 + I_1 - \frac{V_T}{R_\Sigma} W_0 \bigg( I_1 \frac{R_\Sigma}{V_T}  \exp \Big[\frac{R_\Sigma}{V_T} (j^\text{s} + j^\text{a}_1 + I_1) \Big] \bigg)\bigg\}^2$\\
		\hline
		\makecell[l]{Multi-junction cells} & $P_\text{harv}(j^\text{s}; \boldsymbol{j}^{\text{a}}) = \frac{V_T^2 R_\text{L}}{ R_\Sigma^2} \bigg\{ \ln \Big[ \frac{ j^\text{s} + j^\text{a}_1 }{I_n} + 1\Big] + \sum_{n=2}^N \ln \Big[ \frac{ j^\text{a}_n }{I_n} + 1\Big] \bigg\}^2$\\
		\hline
	\end{tabular}
	\label{Table:EHmodels}
	\renewcommand{\arraystretch}{1}\vspace*{-10pt}
\end{table*}

To derive analytical EH models, we first assume that for large shunt resistances $R^\text{sh}_n, \forall n,$ the current leakages $i^\text{sh}_n$ are negligible compared to the current flows through the diodes, i.e., $i_n^\text{sh} \ll i^\text{d}_{n, m}$, $\forall n, m$ \cite{Tietze2012}.
Additionally, we assume that the saturation currents of the diodes modelling the photovoltaic cell are identical\footnotemark\hspace*{0pt}, i.e., $I_{n,1} = I_{n,2} = I_n$, $\forall n$ \cite{Mertens2014, Luque2010}.
Thus, we can rewrite (\ref{Eqn:OutputCurrent_Junction}) as follows:
\begin{equation}
	I_n e^{\frac{ v_n }{V_T}}  +I_n e^{\frac{ v_n }{2 V_T}} + i_{\text{EH}} - j_n - 2 I_n = 0.
	\label{Eqn:OutCurrentQuadratic}
\end{equation}

\footnotetext{If the saturation currents are not identical, one can adopt the current $I_{n} \neq I_{{n},{1}} \neq I_{{n},2}$ that minimizes the expression $ \int_{v_n \in \mathcal{V}} |I_{{n}, 1} e^{\frac{ v_n }{V_\text{T}}} + I_{{n}, 2} e^{\frac{ v_n }{2V_\text{T}}} - I_{n} (e^{\frac{ v_{n} }{V_\text{T}}} + e^{\frac{ v_{n} }{2V_\text{T}}})| \text{d} v_{n}$ for the operating range $v_{n} \in \mathcal{V}$ of the output voltages of the p-n junction $n$.}

Next, we obtain the output voltage $v_n$ of p-n junction $n$ as solution of (\ref{Eqn:OutCurrentQuadratic}), which is a quadratic equation in $e^{\frac{ v_n }{2V_T}}$, as follows\footnotemark:
\begin{align}
	v_n &= \;2 V_T \ln \big[ \frac{ - I_n + \sqrt{ {I_n}^2 - 4 I_n (i_{\text{EH}} - j_n - 2 I_n )  } }{ 2 I_n } \big] \nonumber \\
	&= 2 V_T \ln \bigg[ -\frac{1}{2} + \frac{ 1}{2} \sqrt{ 1 + 4\frac{ j_n - i_{\text{EH}} + 2 I_n  }{I_n}  } \bigg] \nonumber \\
	&\approxeq V_T \ln \bigg[  \frac{ j_n - i_{\text{EH}} }{I_n} + 1 \bigg].
	\label{Eqn:SolQuadEqn}
\end{align}
\footnotetext{Here, the last approximation is based on $\sqrt{x + 2.25} - 0.5 \approx \sqrt{x + 1}, \forall x \geq 0$, which can be verified using any numerical computation tool.}
\noindent Then, we substitute (\ref{Eqn:SolQuadEqn}) into (\ref{Eqn:OutputCurrent}) and obtain:
\begin{align}
	\prod_{n=1}^N I_n = \prod_{n=1}^N ( j_n - i_{\text{EH}} + I_n ) \exp \Big[ - \frac{R_\Sigma}{V_T} i_\text{EH} \Big].
	\label{Eqn:ExpOutputCurrent}
\end{align}
Thus, for any given $j^\text{s}$ and vector of currents $\boldsymbol{j}^\text{a}$, the harvested power is given by (\ref{Eqn:PharvCurrent}), where $i^*_\text{EH}$ can be obtained numerically as solution of (\ref{Eqn:ExpOutputCurrent}).

We note that in contrast to (\ref{Eqn:OutputCurrent}), the solution of (\ref{Eqn:ExpOutputCurrent}) does not require the derivation of inverse functions, and thus, is computationally less demanding than solving (\ref{Eqn:OutputCurrent}).
However, the output current $i^*_\text{EH}$ and harvested power $P_\text{harv}$ still can not be determined analytically in closed form for any $\boldsymbol{j}^\text{a} \in \mathbb{R}_{+}^N$ and $j^\text{s} \geq 0$.
Therefore, in the following, we consider the two special cases of single p-n junction photovoltaic cells and multi-junction cells in the high received power regime.
For both cases, we derive the harvested powers\footnotemark $P_\text{harv}(j^\text{s}; \boldsymbol{j}^\text{a})$ in closed form.
\footnotetext{In the following, we adopt the notation $P_\text{harv}(j^\text{s}; \boldsymbol{j}^\text{a})$ for the harvested power to highlight that it is a function of $j^\text{s}$ and $\boldsymbol{j}^\text{a}$. }

		\paragraph{Single-Junction Photovoltaic Cells}
		\label{Section:ApproximateSJ}
		
We note that in the case of a single-junction photovoltaic cell, i.e., for $N = 1$, (\ref{Eqn:ExpOutputCurrent}) can be equivalently rewritten as follows:
\begin{align}
	I_1 &= ( j_1 - i_{\text{EH}} + I_1 ) \exp \Big[ - \frac{R_\Sigma}{V_T} i_\text{EH} \Big] \nonumber \\
	&=  ( j_1 - i_{\text{EH}} + I_1 ) \exp \bigg[ \frac{R_\Sigma}{V_T} (j_1 - i_\text{EH} + I_1 )\bigg]  \nonumber\\&
	\quad\quad\quad\quad\quad\quad\quad\quad \exp \bigg[- \frac{R_\Sigma}{V_T} (j_1 + I_1) \bigg].
	\label{Eqn:ExpOutputCurrent_SingleJunction}
\end{align}
Therefore, the output current $i^*_\text{EH}$ as solution of (\ref{Eqn:ExpOutputCurrent_SingleJunction}) is given by
\begin{align}
	i^*_{\text{EH}} &= j_1 + I_1 - \frac{V_T}{R_\Sigma} W_0 \bigg( I_1 \frac{R_\Sigma}{V_T}  \exp \bigg[\frac{R_\Sigma}{V_T} (j_1 + I_1) \bigg] \bigg),
	\label{Eqn:EHCurrent_SingleJunction}
\end{align}
\noindent where $W_0(\cdot)$ is the principal branch of the Lambert-W function.
Thus, the harvested power of a single-junction photovoltaic cell can be obtained in closed form as a function of photovoltaic current $j^\text{s}$ as follows:
\begin{align}
	P_\text{harv}(j^\text{s}; &j^\text{a}_1) = R_\text{L} \bigg\{j^\text{a}_1 + j^\text{s} + I_1- \frac{V_T}{R_\Sigma}  \nonumber\\ & W_0 \bigg( I_1 \frac{R_\Sigma}{V_T} \exp \Big[\frac{R_\Sigma}{V_T} (j^\text{a}_1 +  j^\text{s} + I_1) \Big] \bigg)\bigg\}^2.
	\label{Eqn:EHModelSingleJunction}
\end{align}

The derived EH model in (\ref{Eqn:EHModelSingleJunction}) characterizes the harvested power at the optical RX for any values of photovoltaic currents $j^\text{s}$ and $j^\text{a}_1$ induced by the received information FSO signal and the optical energy-providing signal and the ambient light, respectively.
Furthermore, as $j = j^\text{s}  + j^\text{a}_1  \to \infty$, (\ref{Eqn:EHModelSingleJunction}) converges asymptotically to the EH model derived in \cite{Wang2015}, where the EEC of the photovoltaic RX comprises only one diode, which models the diffusion current, while the recombination of charges in the depletion region of the p-n junction is neglected.

		\paragraph{Multi-Junction Photovoltaic Cells}
		\label{Section:ApproximateMJ}

In the following, we consider a broadband multi-junction cell, i.e., $N\geq1$.
Since the photovoltaic RX is illuminated by broadband ambient light and an energy-providing optical signal, we assume that the received powers at all p-n junctions are high \cite{Luque2010}. 
In this case, the current flows through the diodes of the EEC of the p-n junctions are high, and thus, we have $ i_\text{EH} \ll j_n, \forall n,$ and hence, ${j_n - i_\text{EH}} \approxeq {j_n}, \forall n$. 
Then, the output current $i^*_\text{EH}$ as solution of (\ref{Eqn:ExpOutputCurrent}) can be obtained as follows \cite{Tietze2012}:
\begin{align}
	i^*_\text{EH} &= \frac{V_T}{R_\Sigma} \sum_{n=1}^N \ln \Big[ \frac{j_n}{I_n} + 1\Big].
\end{align}
Finally, the harvested power as a function of the FSO photovoltaic current $j^\text{s}$ is given by:
\begin{align}
	P_\text{harv}(j^\text{s} ; \boldsymbol{j}^{\text{a}}) =& \frac{V_T^2 R_\text{L}}{ R_\Sigma^2} \bigg\{  \ln \Big[ \frac{ j^\text{s}  + j^\text{a}_1 }{I_n} + 1\Big] \nonumber \\
	&+ \sum_{n=2}^N \ln \Big[ \frac{ j^\text{a}_n }{I_n} + 1\Big] \bigg\}^2.
	\label{Eqn:EhModelMultiJunction}
\end{align}
The proposed accurate and approximate EH models are summarized in Table~\ref{Table:EHmodels}.
We will exploit the analytical closed-form expressions $P_\text{harv}(j^\text{s}; \boldsymbol{j}^{\text{a}})$ in (\ref{Eqn:EHModelSingleJunction}) and (\ref{Eqn:EhModelMultiJunction}) obtained for photovoltaic RXs equipped with single and multiple p-n junctions, respectively, for the design of the information RX in Section~\ref{Section:ID}.
		
\section{Information Reception at the RX}
	\label{Section:ID}
	In the following, we discuss the information reception at the photovoltaic RX.
To this end, we first show that adopting the optimal filter for the output information signal may not be practical for non-linear photovoltaic RXs, and therefore, we design a suboptimal information RX.
Next, we derive the optimal transmit signal distribution maximizing the achievable information rate, and for practical OOK modulation at the information TX, we obtain the bit-error probability for maximum likelihood detection at the RX. 
	\subsubsection{Information RX}
	\label{Section:Filtering}
	
\begin{table*}[!t]
	\centering
	\caption{Simulation parameters.}
	\begin{tabular}{|c|l|c|c}
		\hline
		\multicolumn{2}{|c|}{\makecell{Parameters of the photovoltaic RXs \cite{Dimroth2016}}} & \multicolumn{2}{c|}{\makecell{Parameters of the TXs and FSO channels}}\\
		\hline		
		\multicolumn{1}{|c|}{\makecell{$N=1$}} & \multicolumn{1}{c|}{ \makecell{ $\lambda^1_\text{min} = \SI{400}{\nano\meter}$, $\lambda^1_\text{max} = \SI{700}{\nano\meter}$, $A_\text{P} = \SI{1}{\centi\meter\squared}$}} & \multicolumn{2}{c|}{\makecell[c]{$\lambda_0 = \SI{980}{\nano\meter}$, $\lambda_n = \frac{\lambda^n_\text{min} + \lambda^n_\text{max}}{2}, n \in\mathbb{N}$}}\\
		\hline		
		\multicolumn{1}{|c|}{\makecell{$N=4$}} & \multicolumn{1}{c|}{ \makecell{ $\lambda^1_\text{min} = \SI{400}{\nano\meter}$, $\lambda^2_\text{min} = \SI{650}{\nano\meter}$, 
				$\lambda^3_\text{min} = \SI{900}{\nano\meter}$, \\  $\lambda^4_\text{min} = \SI{1100}{\nano\meter}$, $\lambda^4_\text{max} = \SI{1800}{\nano\meter}$, $A_\text{P} = \SI{1}{\centi\meter\squared}$ }} & \multicolumn{1}{c|}{\makecell{Channel gains \\ Noise variance}} & \multicolumn{1}{c|}{ \makecell{$h = g_{n} = 1, \forall n $ \\ $\sigma^2 = \SI{-60}{\dBm}$}} \\
		\hline
	\end{tabular}
	\label{Table:SimulationSetup}\vspace*{-10pt}
\end{table*}
As in \cite{Tran2019, Wang2015}, we utilize the current flow $i_\text{ID}(t)$ across the resistance $R_\text{d}$ shown in Fig.~\ref{Fig:MjCell} as information-carrying signal and employ it for information decoding.
Since only the \gls*{ac} component of the received FSO signal is utilized for information reception, the current flow $i_\text{ID}(t)$ can be expressed as \cite{Horowitz1989}:
\begin{equation}
	i_\text{ID}(t) =  i_\text{ID}^\text{s} (t) + i_\text{ID}^\text{n} (t),
	\label{Eqn:IdCurrent}
\end{equation}
where $i_\text{ID}^\text{s} (t)$ is the output current\footnotemark\hspace*{0pt} due to the transmit information signal $s(t)$ and depends on the parameters of the RX and the power spectral density of the received light $p_\text{r} (\lambda, t)$ in (\ref{Eqn:ReceivedSpectrum}) \cite{Horowitz1989}.
\footnotetext{We note that the current flow through a charging (discharging) capacitor can be modelled by exponential function $i_\text{ID}^\text{s} (t) = \pm I_\text{ID}^\text{s} \exp (-\frac{t}{R_\text{eq} C_\text{d}})$, where the sign determines the direction of the current flow and $I_\text{ID}^\text{s}$ and $R_\text{eq}$ are the maximum charging (discharging) current and the equivalent resistance which depend on the parameters of the optical RX, received signal power, and ambient light intensity \cite{Horowitz1989}.}
Furthermore, $i_\text{ID}^\text{n} (t)$ is the equivalent noise at resistance $R_\text{d}$ that comprises the contributions of both the received noise $\tilde{w}(\lambda, t)$ and the noise generated by the elements of the photovoltaic RX \cite{Mertens2014, Wang2015}.
For demodulation of the transmitted message, one can design a filter that is matched to the information signal and yields the maximum achievable output \gls*{snr} after sampling \cite{Horowitz1989, Tse2005}.
However, the design and practical realization of such a filter may not be feasible due to the non-linearity of the photovoltaic RX, i.e., the non-linear dependency of $i_\text{ID}^\text{s} (t)$ on the spectral density of the received light $p_\text{r}(\lambda, t)$, see Table~\ref{Table:EHmodels}, and the high required sampling rate for an accurate discretization of the output signal $i_\text{ID}(t)$ \cite{Horowitz1989, Wang2015}.
Therefore, in the following, we resort to a suboptimal information RX.

To this end, we integrate the voltage $v_\text{d}(t) = i_\text{ID}(t) R_\text{d}$ across resistance $R_\text{d}$ over time slot $k, k\in\mathbb{N},$ and obtain:
\begin{align}
	r[k] &=  \int_{(k-1)T}^{kT} R_\text{d} i_\text{ID}(t) \text{d} t = \int_{(k-1)T}^{kT} R_\text{d} C_\text{d} \frac{\text{d} v_\text{C}(t)  }{\text{d} t} \text{d} t
	\nonumber \\ &= \tilde{x}[k] - \tilde{x}[k-1] + \tilde{n}[k] - \tilde{n}[k-1],
	\label{Eqn:FilterOutputGeneral}
\end{align}
\noindent where $v_\text{C}(t) = v^\text{s}_\text{C}(t) + v^\text{n}_\text{C}(t)$ is the voltage across the capacitance $C_\text{d}$ of the high-pass filter, as shown in Fig.~\ref{Fig:MjCell},~and $v^\text{s}_\text{C}(t)$ and $v^\text{n}_\text{C}(t)$ are the components of $v_\text{C}(t)$ corresponding to the received FSO signal $h s(t)$ and noise $i_\text{ID}^\text{n} (t)$, respectively.
Furthermore, in (\ref{Eqn:FilterOutputGeneral}), $\tilde{x}[k] = R_\text{d} C_\text{d} v^\text{s}_\text{C}(kT)$ and $\tilde{n}[k] = R_\text{d} C_\text{d} v^\text{n}_\text{C}(kT)$ are the information signal and the noise in time slot $k$ at the output of the RX, respectively.

To compute $\tilde{x}[k]$, we observe that in the steady state, the current flow $i_\text{ID}(kT) \approx 0, \forall k$, as shown Fig.~\ref{Fig:InpulseResponse}, and voltage $v_\text{C}(kT)$ is equal to the voltage $R_\text{L} i_\text{EH}(kT)$ across the EH load resistance $R_\text{L}$ in time slot $k$ \cite{Horowitz1989}, see Fig.~\ref{Fig:MjCell}.
Thus, we obtain the information symbol received in time slot $k, \forall k,$ as 
\begin{align}
	\tilde{x}[k] &= R_\text{d}C_\text{d} v^\text{s}_\text{C}(kT) \nonumber \\
	&= R_\text{d}C_\text{d} \sqrt{R_\text{L}} \sqrt{P_\text{harv}(j^\text{s}[k]; \boldsymbol{j}^\text{a} )},
\end{align}
where $j^\text{s}[k] = h s[k] r_1(\lambda_0)$ and $P_\text{harv}(\cdot; \cdot)$ is the power harvested at the EH load that is derived in Section~\ref{Section:EH} and given in Table~\ref{Table:EHmodels}.

We note that the output symbol $r[k], \forall k,$ in (\ref{Eqn:FilterOutputGeneral}) depends not only on $\tilde{x}[k]$ but also on $\tilde{x}[k-1]$.
To avoid the undesired memory of the photovoltaic RX, we assume that a sequence of symbols $\boldsymbol{s} = \{s[0], s[1], \cdots, s[K-1]\}$ of length $K \in \mathbb{N}$ is transmitted.
Then, we obtain the normalized output symbol $y[k]$ in time slot $k, \forall k$, as follows\footnotemark:
\begin{equation}
	y[k] = \frac{1}{R_\text{d}C_\text{d} \sqrt{R_\text{L}}} \sum_{p=0}^k r[p] = x[k] + n[k],
	\label{Eqn:CommChannel}
\end{equation}
where $x[k] = \sqrt{P_\text{harv}(j^\text{s}[k]; \boldsymbol{j}^\text{a} )}$ and $n[k] = \frac{1}{ R_\text{d}C_\text{d} \sqrt{R_\text{L}} } \tilde{n}[k]$ are the normalized output information signal and noise, respectively.
\footnotetext{We note that if the information RX is able to measure voltage $v_\text{C}(t)$, the integration of the current flow $i_\text{ID}(t)$ in (\ref{Eqn:FilterOutputGeneral}) is not needed and the output symbol can be directly obtained as $y[k] = R_\text{L}^{-\frac{1}{2}} v_\text{C}(kT), \forall k$.}
Thus, unlike for optical RXs based on photodetectors in \cite{MaShuai2019, Ajam2020}, the output information signal $x[k]$ at the photovoltaic RX is not a linear function of the TX symbol $s[k], \forall k,$ but is determined by the non-linear functions in Table~\ref{Table:EHmodels}.

We note that the output noise samples ${n}[k], \forall k \in\mathbb{N},$ include the impacts of the received noise as well as the thermal and shot noise of the photovoltaic RX \cite{Lapidoth2009}.
To model ${n}[k], \forall k$, we assume that the thermal noise at the RX dominates and the impact of transmit information symbol $s[k]$ on $n[k]$ is negligible \cite{Wang2015}.
Thus, we model ${n}[k], \forall k,$ as \gls*{iid} realizations of AWGN with variance $\sigma^2$ \cite{Lapidoth2009, Wang2015}.
	\subsubsection{Achievable Information Rate}
	\label{Section:ContDistribution}
	
The considered photovoltaic RX harvests power from ambient light, an optical energy-providing signal, and an information FSO signal.
However, information transmission is challenging and constitutes a performance bottleneck of practical SLIPT systems \cite{Sepehrvand2021, Wang2015, Fakidis2020}.
Therefore, in this paper, we design the optical SLIPT system assuming that for a given vector $\boldsymbol{j}^\text{a}$, the information FSO TX aims at maximizing the information rate between the transmit and receive information signals, while the RX opportunistically harvests the powers of the ambient light and the received optical energy-providing and information FSO signals.
In the following proposition, for a given maximum information transmit power $A^2$ and power spectral density of the optical energy-providing signal $p_\text{e}(\lambda)$, we determine the optimal distribution of transmit symbol $s$, which maximizes the achievable rate of the considered SLIPT system employing a multi-junction photovoltaic RX.
\begin{proposition}
	\label{Theorem:Capacity}
	For a given cdf of the transmit information symbol, $F_{s}(s)$, the achievable rate in nats per channel use is given by
	\begin{equation}
		R(F_{s}) = \frac{1}{2} \ln \left[ 1+ \frac{e^{2u(x; F_{s})}}{ 2\pi e \sigma^2 } \right],
		\label{Eqn:RateBoundGeneral}
	\end{equation}
	\noindent where $u(x; F_{s})$ is the differential entropy of random variable $x = \sqrt{P_\text{\upshape harv}(j^\text{\upshape s}[k]; \boldsymbol{j}^\text{\upshape a})}$ for the given $F_{s}(\cdot)$.
	Furthermore, the achievable rate in (\ref{Eqn:RateBoundGeneral}) is maximized by the following cdf of the transmit information symbol
	\begin{equation}
		F^*_{s}(s) = \begin{cases}
			0, \quad & \text{\upshape if} \, s<0, \\
			\frac{ \sqrt{P_\text{\upshape harv}( h s r_1(\lambda_0); \boldsymbol{j}^\text{\upshape a} )} -  x_0  }{ \theta(A^2) } , \quad &\text{\upshape if} \, s\in[0, A^2], \\
			1, \quad & \text{\upshape if} \, s>A^2,
		\end{cases}
		\label{Eqn:CapacityAchievingDistributionS}
	\end{equation}
	\noindent where $x_{A} = \sqrt{P_\text{\upshape harv}( h A^2 r_1(\lambda_0); \boldsymbol{j}^\text{\upshape a})}$, $x_0 = \sqrt{P_\text{\upshape harv}( 0; \boldsymbol{j}^\text{\upshape a})} $, and $\theta(A^2) = x_A - x_0$, and can be expressed as a function of $A^2$ as follows:
	\begin{equation}
		\bar{R}(A^2) = \frac{1}{2} \ln \Big( 1 + \frac{ \theta^2(A^2)  } {2 \pi e \sigma^2 } \Big).
		\label{Eqn:RateBound}
	\end{equation}
\end{proposition}
\begin{proof}
	Please refer to Appendix~\ref{Appendix:ProofCapacity}.
\end{proof}

Proposition~\ref{Theorem:Capacity} shows that in contrast to linear optical information RXs based on photodetectors \cite{MaShuai2019, Ajam2020}, uniformly distributed information transmit symbols $s$ are not optimal for non-linear photovoltaic RXs.
We note that the cdf $F^*_{s}(\cdot)$ in (\ref{Eqn:CapacityAchievingDistributionS}), and thus, the achievable information rate in (\ref{Eqn:RateBound}) depend on the EH model in Table~\ref{Table:EHmodels}.
Furthermore, the achievable information rate in Proposition~\ref{Theorem:Capacity} does not depend on the instantaneous output power at the RX, but is determined by the sensitivity $\theta(A^2)$ of the photovoltaic RX.
In the following corollary, we determine the average harvested power at the RX when the distribution of transmit symbols $s$ maximizing the achievable rate is adopted at the information FSO TX.

\begin{corollary}
	\label{Theorem:PowerCont}
	For the cdf of transmit symbol $s$ in (\ref{Eqn:CapacityAchievingDistributionS}), the average harvested power at the photovoltaic RX as a function of $A^2$ is given by
	\begin{equation}
	\bar{P}_\text{\upshape harv}(A^2) = \frac{1}{3} ( x^2_{A} + x^2_0 + {x_{A}  x_0 } ),
	\label{Eqn:ContAverageHarvestedPower}
\end{equation} 
\end{corollary}
\begin{proof}
	Please refer to Appendix~\ref{Appendix:ProofPowerCont}.
\end{proof}

Corollary~\ref{Theorem:PowerCont} reveals that since $x_0 > 0$ for $j^\text{a}_n > 0, \forall n$, both the achievable information rate and the corresponding average harvested power depend not only on cdf $F_{s}$, but also on the power spectral densities ${p}_\text{a}(\lambda)$ and ${p}_\text{e}(\lambda)$ of the ambient light and the optical energy signal.
Thus, an accurate modelling of optical RXs is important to characterize both the harvested power and the performance of the information RX.

In the following, since the optimal distribution of transmit information signal $s$ in Proposition~\ref{Theorem:Capacity} may be impractical for SLIPT systems, we assume that OOK modulation is adopted at the information TX and derive the resulting bit-error probability for maximum likelihood information detection at the photovoltaic RX.

\begin{figure}[!t]
	\centering
	\includegraphics[draft=false, width=0.44\textwidth]{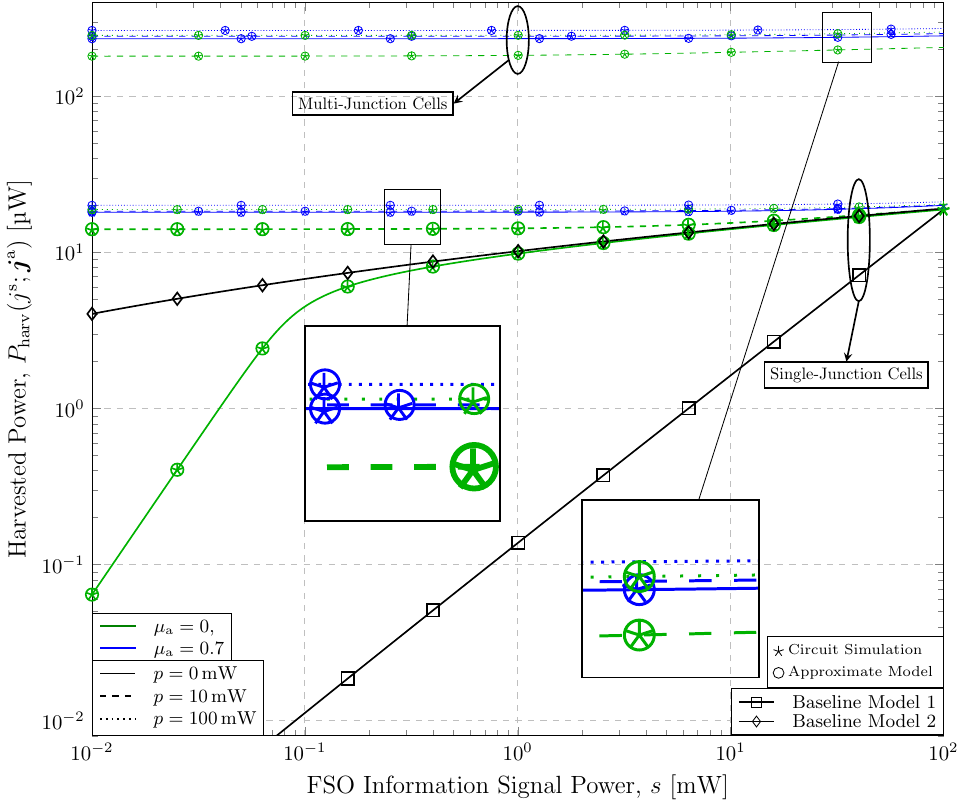} \vspace*{-7pt}
	\caption{Validation of the derived closed-form analytical EH models via circuit simulations for photovoltaic RXs equipped with single ($N = 1$) and multiple ($N=4$) p-n junctions for different ambient light intensities $\mu_a$ and FSO TX powers $s$ and $p$.}
	\label{Fig:ResultsValidationEH}	\vspace*{-7pt}
\end{figure}

	\subsubsection{OOK Modulation}
	\label{Section:OokModulation}
	
In the following, we assume that the information FSO TX employs \gls*{ook} modulation with equiprobable symbols, i.e., the pdf of $s$ is binary with $\text{dom} \{f_{s}\} = \{0,A^2\}$ and $f_{s}(0) = f_{s}(A^2) = \frac{1}{2}$.
In this case, we have $x[k] \in \{ x_0, x_{A}  \}$ and the maximum likelihood detected signal $\hat{s}[k]$ in time slot $k$ at the photovoltaic RX can be expressed as follows \cite{Ajam2020}:
\begin{equation}
	\hat{s}[k] = \begin{cases}
		A^2, \quad \text{if} \; y[k] \geq \frac{1}{2} (x_0 + x_{A}), \\
		0, \;\;\quad \text{otherwise}.
	\end{cases}
\label{Eqn:MLDecoding}
\end{equation}

In the following proposition, we determine the bit-error probability for maximum likelihood signal detection at the photovoltaic RX.
\begin{proposition}
	The bit-error probability for maximum likelihood detection at the RX as function of $A^2$ is given by
	\begin{equation}
		P_\text{\upshape e}(A^2) = Q\bigg(\frac{\theta(A^2)}{ 2\sigma} \bigg),
		\label{Eqn:BER}
	\end{equation}
	\noindent where $Q(\cdot)$ denotes the Gaussian Q-function.
	\label{Theorem:BER}
\end{proposition} 
\begin{proof}
	Please refer to Appendix~\ref{Appendix:ProofBER}.
\end{proof}
Thus, we conclude that similar to the maximum achievable information rate in (\ref{Eqn:RateBound}) and the average harvested power in (\ref{Eqn:ContAverageHarvestedPower}), the bit-error probability at the optical photovoltaic RX also depends not only on $A^2$, but also on the spectral power densities of the received ambient light and the optical energy signal.
In the following section, we will numerically investigate the performance of EH and information reception at the considered photovoltaic RX for different ambient light intensities, powers of the energy-providing transmit signal, and maximum FSO TX powers $A^2$.

\section{Numerical Results}
	\label{Section:NumResults}
	In this section, we evaluate the performance of the proposed SLIPT system numerically.
First, we discuss the setup adopted for our simulations.
Next, we validate the analytical EH models and evaluate the performance of information reception at the photovoltaic RX derived in Sections~\ref{Section:EH} and \ref{Section:ID}, respectively.
Finally, we investigate the rate-power tradeoff for SLIPT systems with multi-junction optical RXs.

	\subsection{Simulation Setup}
	\label{Section:SimSetup}
	In the following, we present the system parameters adopted for our numerical simulations.
To this end, we first discuss the adopted black body model for the power spectrum of the ambient light.
Then, we elaborate on the parameters of the considered SLIPT system.

\subsubsection{Black Body Model of Ambient Light}
In our numerical results, we assume that the photovoltaic cell is illuminated by sunlight, and to model the irradiance spectrum of the sun, we adopt the black body model \cite{Mertens2014}.
The normalized spectral irradiance of the absolute black body follows Planck's law and can be expressed as a function of wavelength $\lambda$ as follows \cite{Mertens2014}:
\begin{equation}
	p_\text{B}(\lambda) = \frac{2 k_\text{p} c^2}{\lambda^5} \frac{1}{\exp(\frac{k_\text{p}  c}{k_\text{b} T_\text{sun} \lambda}) - 1},
	\label{Eqn:PlanckLaw}
\end{equation}
\noindent where $c$ is the speed of the light, $k_\text{p}$ and $k_\text{b}$ are the Planck and Boltzmann constants, respectively, and $T_\text{sun} = \SI{5778}{\kelvin}$ is the equivalent absolute temperature of the sun.
As a result, the power spectral density of the ambient light is given by $	p_\text{a}(\lambda) = \mu_a \nu_\text{s} p_\text{B}(\lambda)$, where $\mu_a\in[0,1]$ and $\nu_\text{s} = \frac{ \alpha_\text{SE} A_\text{S} A_\text{P} }{A_\text{E}}$ are the sunlight intensity at the RX and the equivalent size of the cell, respectively \cite{Mertens2014}.
Here, $\alpha_\text{SE} = 5.72 \cdot 10^{-9} \SI{}{\steradian}$ is the solid angle of the Earth seen from the sun, $A_\text{S} = 6.07\cdot 10^{12}\, \SI{}{\square\kilo\meter}$ and $A_\text{E} = 5.1\cdot 10^{8}\, \SI{}{\square\kilo\meter}$ are the areas of the Earth and sun surface, respectively, and $A_\text{P}$ is the surface area of the photovoltaic cell \cite{Luque2010}.

Coefficient $\mu_a$ is used to capture atmospheric losses and misalignment between the photovoltaic cell orientation and the direction of the sun \cite{Mertens2014, Luque2010}.
In particular, $\mu_a = 0$ and $\mu_a = 1$ correspond to the two cases, where the sunlight has no effect on the harvested power and where the atmosphere is transparent and the photovoltaic cell is perfectly aligned, respectively.
Furthermore, we note that in indoor scenarios, the normalized spectral irradiance of a lightbulb or an illuminating lamp can also be approximated by Planck's law (\ref{Eqn:PlanckLaw}) and coefficient $\mu_\text{a} \in \mathbb{R}_{+}$ can be chosen to yield a match between the power spectral density $p_\text{a}(\lambda)$ and measurement or simulation results \cite{Luque2010, Mertens2014}.

\subsubsection{Simulation Parameters}
For our simulations, we adopt photovoltaic cells of size $A_\text{P} = \SI{1}{\centi\meter\squared}$ equipped with a single ($N=1$) and $N=4$ p-n junctions modelled by the EEC shown in Fig.~\ref{Fig:MjCell}, whose parameters are specified in Table~\ref{Table:SimulationSetupEEC}.
We assume that the passbands of the p-n junctions of the photovoltaic cell do not intersect, i.e., the responsivity of p-n junction $n$ is modelled as follows \cite{Mertens2014}:
\begin{equation}
	\mu_n(\lambda) = \begin{cases}
		0.7\frac{q_0}{k_\text{p} c}, &\text{if} \; [\lambda^n_\text{min}, \lambda^n_\text{max}], \\
		0, &\text{otherwise},
	\end{cases}
\end{equation}
where $q_0$ is the elementary charge and $\lambda^n_\text{max} = \lambda^{n+1}_\text{min}$, $n\in\{1,2, \cdots, N-1\}$.
We specify the adopted parameters of the considered photovoltaic RXs in Table~\ref{Table:SimulationSetup}.

Next, for the FSO system, we set the carrier wavelength of the information-carrying signal to $\lambda_0 = \SI{980}{\nano\meter}$.
Furthermore, we model the broadband optical energy-providing signal as a sum of $N$ FSO laser signals: 
\begin{equation}
	p_\text{e}(\lambda) = \sum_{n=1}^N p_n \delta(\lambda - \lambda_n),
\end{equation}
where $p_n > 0$ and $\lambda_n$ are the power and wavelength of FSO energy-providing signal $n, n \in \{1,2,\cdots, N\}$, respectively.
Moreover, the channel between the optical energy source and the RX is modelled as $g(\lambda) = \sum_{n=1}^N g_n \delta(\lambda - \lambda_n)$, where $g_n > 0$ is the channel between FSO energy TX $n$ and the RX.
We note that multi-junction photovoltaic cells are not able to operate efficiently if the induced current of any p-n junction is low \cite{Luque2010, Mertens2014}.
Therefore, to enable efficient EH at the RX also when the ambient light intensity is low, i.e., $\mu_\text{a} \approx 0$, we set the carrier wavelength of FSO energy-providing signal $n$ to $\lambda_n = \frac{\lambda^n_\text{min} + \lambda^n_\text{max}}{2}, n \in\{1,2,\cdots, N\}$.
Finally, we set $p_n = \frac{p}{N}, \forall n$, where $p \in \mathbb{R}_{+}$ is the total power of the energy-providing FSO signal independent of the number of p-n junctions $N$, assume normalized channel gains $h = g_n = 1, \forall n,$ and set the AWGN variance to $\sigma^2 = \SI{-60}{\dBm}$.
The adopted simulation parameters are summarized in Tables~\ref{Table:SimulationSetupEEC} and \ref{Table:SimulationSetup}.

	\subsection{Validation of EH Models and Sensitivity of Optical RXs}
	\label{Section:ValidationEH}
	
\begin{figure*}[!t]
	\centering
	\subfigure[Sensitivity for single-junction RXs, $N=1$]{
		\includegraphics[draft=false, width=0.47\textwidth]{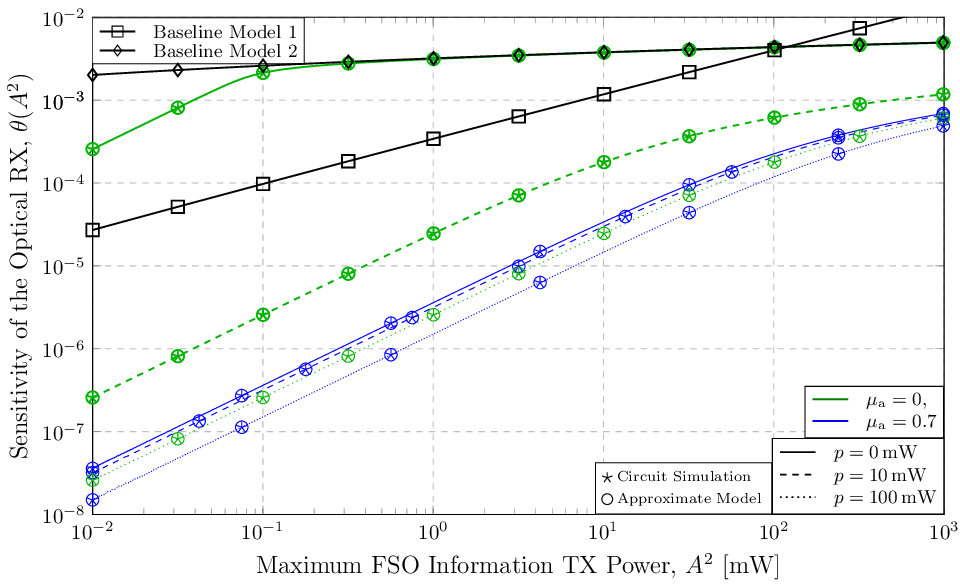} \label{Fig:SensitivitySJ}}
	\subfigure[Sensitivity for multi-junction RXs, $N=4$]{
		\includegraphics[draft=false, width=0.47\textwidth]{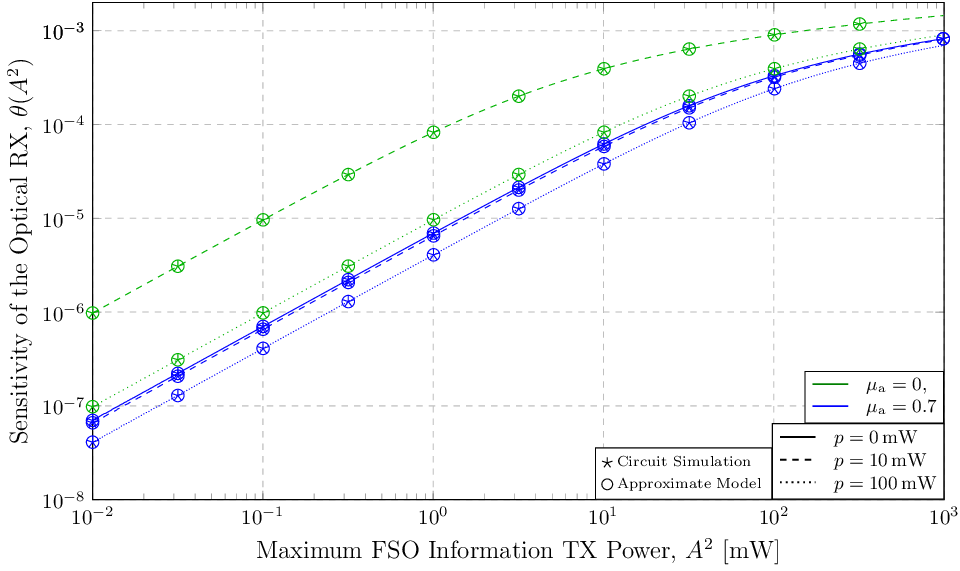}
		\label{Fig:SensitivityMJ}}
	\caption{Sensitivity of the optical RXs $\theta(A^2)$ for different ambient light intensities $\mu_a$ and FSO TX powers $A^2$ and $p$.}
	\label{Fig:ResultsSensitivity}\vspace*{-10pt}
\end{figure*}

\begin{figure*}[!t]
	\centering
	\subfigure[Cdfs for single-junction RXs, $N=1$]{
		\includegraphics[draft=false, width=0.47\textwidth]{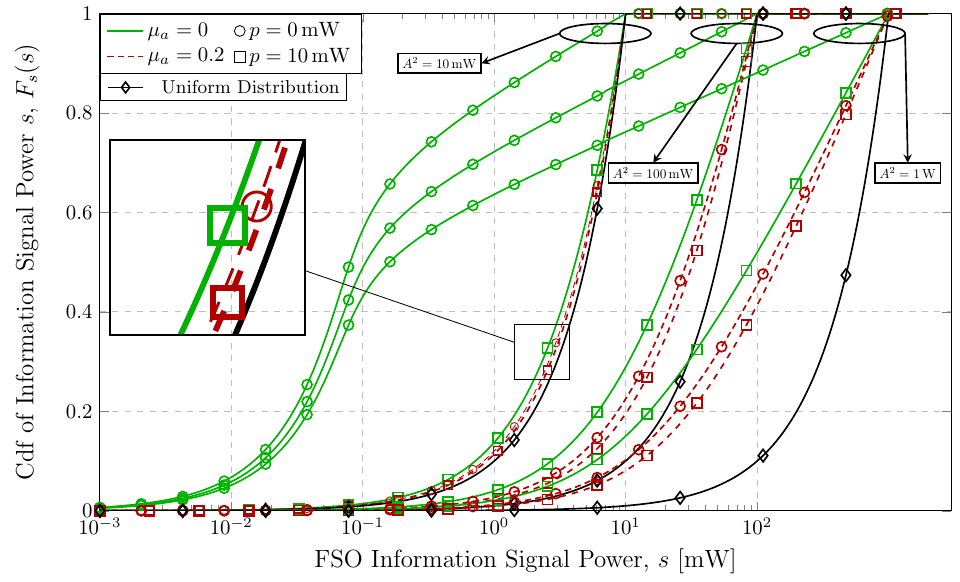} \label{Fig:Cdfs_SJ}}
	\subfigure[Cdfs for multi-junction RXs, $N=4$]{
		\includegraphics[draft=false, width=0.47\textwidth]{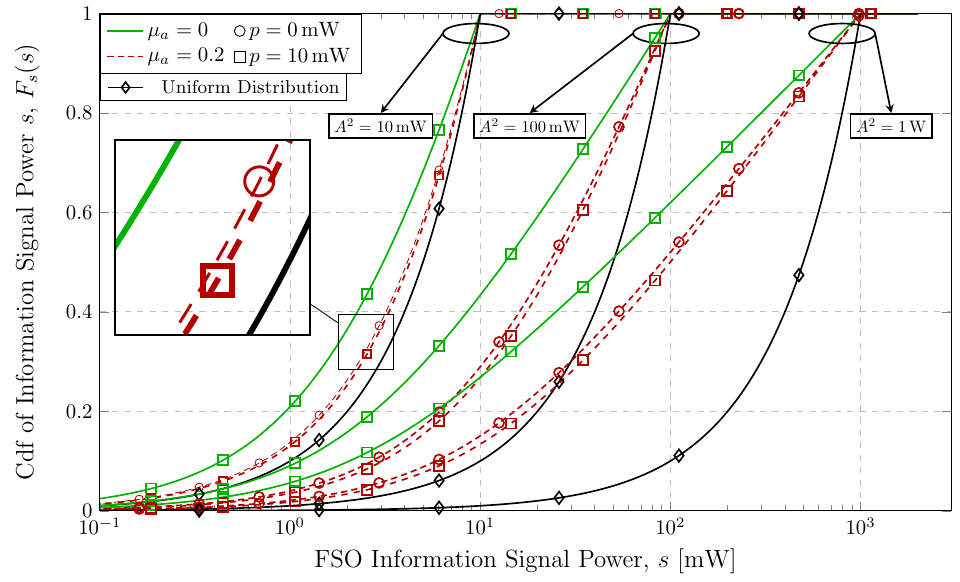}
		\label{Fig:Cdfs_MJ}}
	\caption{Proposed cdfs $F_{s}^*$ in (\ref{Eqn:CapacityAchievingDistributionS}) and cdfs of uniformly distributed information TX signal $s$ for different values of maximum FSO TX powers $A^2$ and $p$.}
	\label{Fig:ResultsCdfs}\vspace*{-10pt}
\end{figure*}

In this section, we validate the closed-form analytical EH models derived in Section~\ref{Section:EH} via circuit simulations.
To this end, in Fig.~\ref{Fig:ResultsValidationEH}, we plot the harvested powers $P_\text{harv}(j^\text{s}; \boldsymbol{j}^\text{a})$ obtained analytically from the expressions summarized in Table~\ref{Table:EHmodels} for single-junction ($N=1$) and multi-junction ($N=4$) photovoltaic RXs for different values of ambient light intensity $\mu_\text{a}$ and FSO transmit powers $s$ and $p$.
Furthermore, for different values of $s, p$, and $\mu_\text{a}$, we also simulate the RX EEC shown in Fig.~\ref{Fig:MjCell} with the circuit simulation tool Keysight ADS \cite{ADS2017} and show the obtained harvested powers as star markers in Fig.~\ref{Fig:ResultsValidationEH}.
To investigate the EH efficiency at the RX for low and high ambient light intensities, we adopt $\mu_\text{a} \in \{0, 0.7\}$.
For single-junction RXs, we adopt FSO energy signal powers $p = \{\SI{0}{\milli\watt}, \SI{10}{\milli\watt}, \SI{100}{\milli\watt}\}$.
Since a broadband high-power received optical signal is assumed for the derivation of the closed-form analytical EH model in Section~\ref{Section:ApproximateMJ}, for multi-junction RXs, we adopt $p = \{\SI{10}{\milli\watt}, \SI{100}{\milli\watt}\}$ and $p = \{\SI{0}{\milli\watt},  \SI{10}{\milli\watt}, \SI{100}{\milli\watt}\}$ for $\mu_\text{a} = 0$ and $\mu_\text{a} = 0.7$, respectively.
Additionally, as Baseline Model 1 and Baseline Model 2, we adopt the EH models derived in \cite{MaShuai2019} and \cite{Wang2015} that assume MPP tracking at the RX and a single-diode EEC of the photovoltaic RX, respectively.
In particular, since the impact of ambient light was neglected for the analysis of SLIPT systems in \cite{MaShuai2019}, we set the model parameters for Baseline Model 1 to match the harvested power $P_\text{harv}(j^\text{s}; \boldsymbol{j}^\text{a})$ in (\ref{Eqn:EHModelSingleJunction}) for $s = \SI{100}{\milli\watt}$, $p = \SI{0}{\milli\watt}$, and $j^\text{a}_1 = \SI{0}{\ampere}$.

First, in Fig.~\ref{Fig:ResultsValidationEH}, we observe that for both photovoltaic RX designs and for all considered values of $\mu_\text{a}, p$, and $s$, the derived closed-form analytical EH models $P_\text{harv}(\cdot; \cdot)$ match the ADS circuit simulation results well.
Next, we note that as expected, for both RX designs, the harvested power at the RX grows with the ambient light intensity $\mu_\text{a}$ and FSO transmit signal powers $p$ and $s$.
Furthermore, we observe that due to a larger passband, the multi-junction RX is able to achieve significantly higher efficiencies for harvesting the received broadband optical signal.
Also, since the EH model in \cite{MaShuai2019} assumes a tunable EH load resistance $R_\text{L}$ to maximize the EH efficiency for a given FSO transmit power $p$, we observe in Fig.~\ref{Fig:ResultsValidationEH} that for single-junction RXs and $\mu_\text{a}= 0$, Baseline Model 1 is not able to capture the non-linearities of the photovoltaic RX since the parameters of the EH load are fixed.
Finally, for the EH model in \cite{Wang2015}, the EEC comprises only one diode to model the diffusion current of the p-n junction at the RX.
Thus, Baseline Model 2 can not accurately capture the RX non-linearities when the received FSO power is low and the harvested power is determined by the recombination of particles in the depletion region of the junction \cite{Luque2010}.

As we can observe from Proposition~\ref{Theorem:Capacity}, for a given maximum FSO transmit power $A^2$, the performance of information decoding at the optical RX, i.e., the achievable information rate $\bar{R}(\cdot)$ in (\ref{Eqn:RateBound}) and the bit-error probability $P_\text{e}(\cdot)$ in (\ref{Eqn:BER}), are not determined by the instantaneous harvested power at the optical RX, but the sensitivity value $\theta(A^2)$.
Therefore, to investigate the sensitivity of the considered optical photovoltaic RXs to the received information signal, we calculate and plot in Fig.~\ref{Fig:ResultsSensitivity} the values $\theta(A^2)$ for different ambient light intensities $\mu_\text{a}$ and energy-providing FSO signal powers $p$ obtained with the closed-form EH models summarized in Table~\ref{Table:EHmodels}.
Moreover, for the results in Figs.~\ref{Fig:SensitivitySJ} and \ref{Fig:SensitivityMJ}, we assume that the photovoltaic RX is equipped with $N = 1$ and $N = 4$ p-n junctions, respectively.
As in Fig.~\ref{Fig:ResultsValidationEH}, for multi-junction cells with $N=4$, we adopt FSO signal powers $p = [\SI{10}{\milli\watt}, \SI{100}{\milli\watt}]$ and $p = [\SI{0}{\milli\watt}, \SI{10}{\milli\watt}, \SI{100}{\milli\watt}]$ when the ambient light intensity is $\mu_\text{a} = 0$ and $\mu_\text{a} > 0$, respectively.
Finally, similar to Fig.~\ref{Fig:ResultsValidationEH}, we additionally calculate and plot in Fig.~\ref{Fig:ResultsSensitivity} the sensitivities of information decoding $\theta(A^2)$ using Keysight ADS circuit simulations \cite{ADS2017}.

First, we observe in Fig.~\ref{Fig:ResultsSensitivity} that in general, a higher maximum FSO transmit power $A^2$ yields a higher value of the sensitivity function $\theta(A^2)$.
However, for large information FSO transmit powers, the optical RXs tend to saturate.
Furthermore, we observe that in general, higher values of $\mu_a $ and $p$ yield lower optical RX sensitivities $\theta(A^2)$, and thus, lower ambient light intensities and FSO energy signal powers may be preferable for efficient information transmission.
Thus, accurate EH modelling is important for the characterization of both the harvested power at the optical RX and the performance of information transmission in SLIPT systems.
Moreover, we note that in contrast to the baseline EH models, the derived EH models are able to accurately capture not only non-linear, but also saturation effects of the optical photovoltaic RXs.
In the following, we investigate the achievable rates and bit-error probabilities at the optical RXs in detail.

	\subsection{Optimal Distributions, Achievable Rates, and Bit-Error Rates}
	\label{Section:PerformanceId}
	
\begin{figure*}[!t]
	\centering
	\subfigure[Achievable rates for $p = \SI{0}{\milli\watt}$]{
		\includegraphics[draft=false, width=0.48\textwidth]{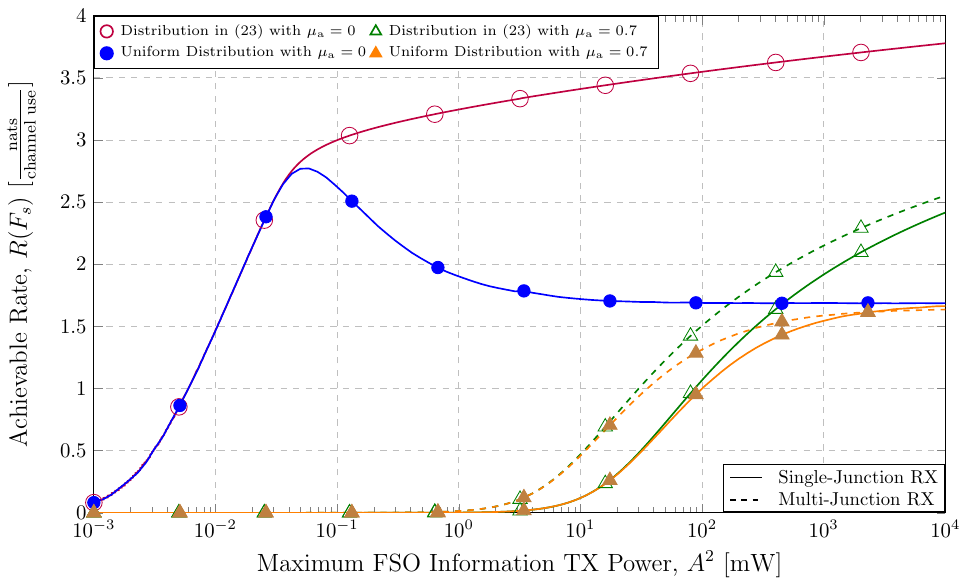}
		\label{Fig:ResultsAchievableRate_s0}}%
	\subfigure[Achievable rates for $p = \SI{100}{\milli\watt}$]{
		\includegraphics[draft=false, width=0.48\textwidth]{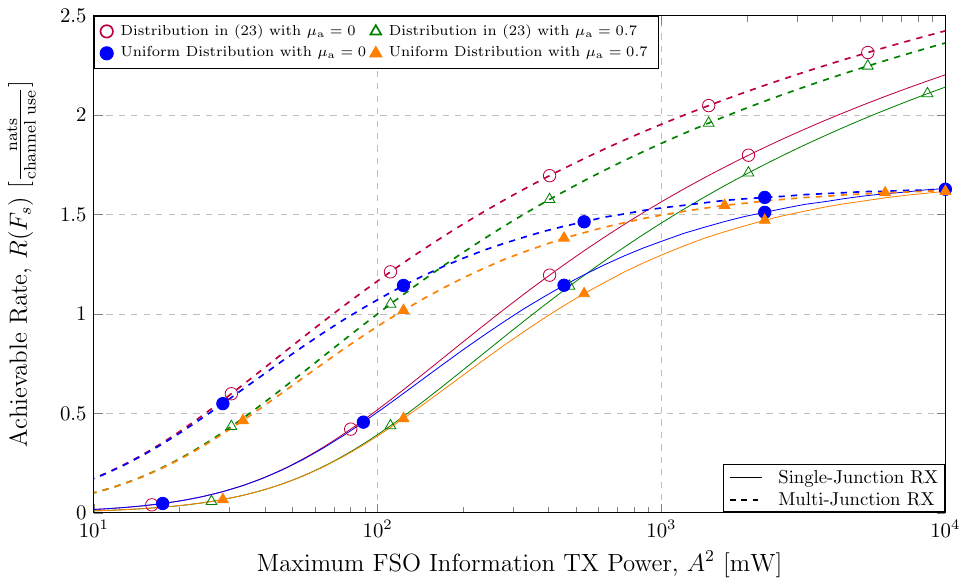} \label{Fig:ResultsAchievableRate_s100}}
	\caption{Achievable information rates $\bar{R}(F_{s})$ for different cdfs $F_{s}$ and different maximum FSO TX signal powers $A^2$ and $p$.}
	\label{Fig:ResultsAchievableRate}\vspace*{-12pt}
\end{figure*}

\begin{figure}[!t]
	\centering
	\includegraphics[draft=false, width=0.48\textwidth]{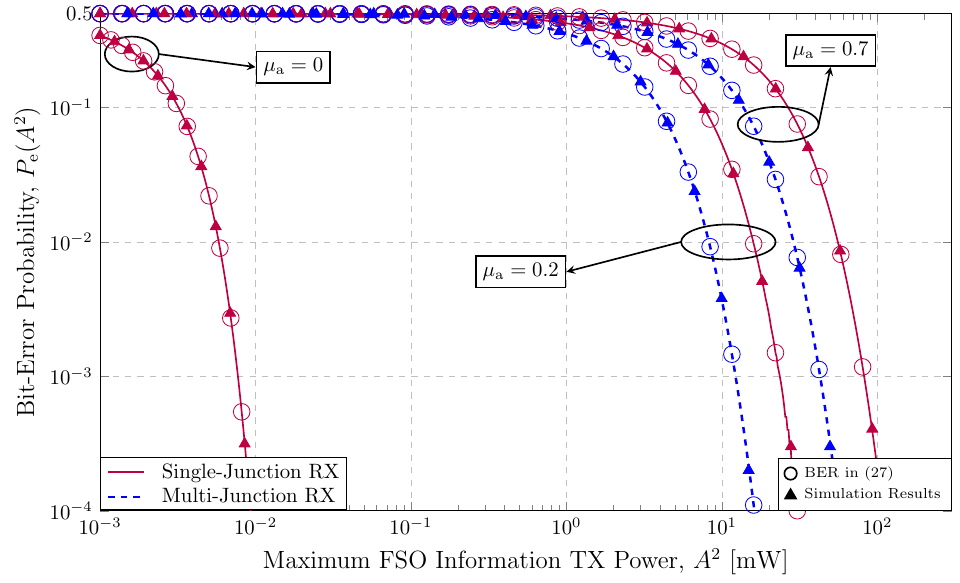}\vspace*{-7pt}
	\caption{Bit-error rates $P_\text{e}$ in (\ref{Eqn:BER}) for different RX designs, intensities $\mu_\text{a}$, and FSO signal power $A^2$.}
	\label{Fig:ResultsBERs}\vspace*{-15pt}
\end{figure}

In this section, we investigate the performance of information reception at non-linear single- and multi-junction photovoltaic RXs.
First, in Fig.~\ref{Fig:ResultsCdfs}, we plot the optimal cdfs $F_{s}^*$ given in (\ref{Eqn:CapacityAchievingDistributionS}), which maximize the achievable information rate $R(F_{s})$ in (\ref{Eqn:RateBoundGeneral}) for different ambient light intensities $\mu_\text{a}$ and maximum powers $p$ and $A^2$ of the FSO energy and information signals, respectively.
In particular, for the cdfs in Figs.~\ref{Fig:Cdfs_SJ} and \ref{Fig:Cdfs_MJ}, we assume that $N=1$ and $N=4$ p-n junctions are adopted at the photovoltaic RX, respectively.
Furthermore, as a baseline scheme, we also plot in Fig.~\ref{Fig:ResultsCdfs} the cdfs that correspond to a uniformly distributed information FSO signal $s \in [0, A^2]$, which maximize the achievable rate of FSO communication systems if, as in \cite{MaShuai2019, Li2017, Pan2019, Diamantoulakis2018, Tran2019, Sepehrvand2021, Papanikolaou2021, Zhang2018, Makki2018, Wang2015, Fakidis2020}, the non-linearity of the photovoltaic RX with respect to the received information signal is neglected.
We observe in Fig.~\ref{Fig:ResultsCdfs} that for both considered RX designs and for all considered values of $\mu_\text{a}, p$, and $A^2$, the proposed cdfs $F_{s}^*$ maximizing the achievable rate of SLIPT systems differ substantially from the cdf for a uniformly distributed information FSO signal $s$ at the TX.
Moreover, we note that the baseline uniform distribution favors higher FSO information signal powers $s$, and thus, may require a significantly higher average power of the transmit information signal than the proposed distribution in (\ref{Eqn:CapacityAchievingDistributionS}).
Thus, the accurate modelling of the RX non-linearities is important for power efficient SLIPT system design.

Next, in Fig.~\ref{Fig:ResultsAchievableRate}, for both considered RX designs and for different ambient light intensities $\mu_\text{a}$ and maximum TX powers $A^2$, we plot the maximum achievable rates in (\ref{Eqn:RateBound}) obtained for the proposed cdf $F_{s}^*$ and the achievable rates in (\ref{Eqn:RateBoundGeneral}) when a uniformly distributed signal $s$ is adopted at the information TX.
For the results in Figs.~\ref{Fig:ResultsAchievableRate_s0} and \ref{Fig:ResultsAchievableRate_s100}, we adopt energy signal powers $p = \SI{0}{\milli\watt}$ and $p = \SI{100}{\milli\watt}$, respectively.
Moreover, for the considered RX designs, in Fig.~\ref{Fig:ResultsBERs}, we also plot the bit-error rates for maximum likelihood information detection at the RX obtained with Proposition~\ref{Theorem:BER}.
Since higher ambient light intensities $\mu_\text{a}$ are equivalent to an increase of the FSO transmit power $p$, for the results in Fig.~\ref{Fig:ResultsBERs}, we adopt $p = 0$ and $\mu_\text{a} = [0, 0.2, 0.7]$ and $\mu_\text{a} = [0.2, 0.7]$ for the RX designs with $N = 1$ and $N = 4$, respectively.
Finally, to validate the bit-error probability in (\ref{Eqn:BER}), we additionally simulate the proposed maximum likelihood information detection scheme at the photovoltaic RX and show the obtained bit-error rates with filled markers in Fig.~\ref{Fig:ResultsBERs}.

First, we observe in Figs.~\ref{Fig:ResultsAchievableRate} and \ref{Fig:ResultsBERs} that, as expected, for both considered RX designs and all considered values of $p$ and $\mu_\text{a}$, the maximum achievable rate $R(F^*_{s})$ and the bit-error rate $P_\text{e}(A^2)$ increases and decreases, respectively, when the maximum transmit information signal power $A^2$ grows.
Furthermore, the achievable data rates obtained in Fig.~\ref{Fig:ResultsAchievableRate} with uniformly distributed TX signals $s$ are not only lower than those of the proposed scheme, but may even decrease with $A^2$.
We note that since the sensitivities in Fig.~\ref{Fig:ResultsSensitivity} saturate for large received signal powers, the higher ambient light intensity $\mu_\text{a} = 0.7$ yields lower achievable data rates in Fig.~\ref{Fig:ResultsAchievableRate} and higher bit-error probabilities in Fig.~\ref{Fig:ResultsBERs}.
Thus, there is a tradeoff between the achievable data rates and the harvested power in SLIPT systems, which will be studied in detail in the next section.
We note that since the efficiency of multi-junction photovoltaic cells is compromised when the received optical signal is not broadband, i.e., $p = 0$ and $\mu_\text{a} = 0$ \cite{Luque2010}, single-junction photovoltaic cells yield superior data rates and bit-error probabilities in this case.
However, in Figs.~\ref{Fig:ResultsValidationEH}, \ref{Fig:ResultsAchievableRate}, and \ref{Fig:ResultsBERs}, we observe that when $\mu_\text{a} > 0$, multi-junction photovoltaic cells are not saturated for the considered values of $p$ and $A^2$ due to the efficient allocation of the optical power to different p-n junctions.
Thus, surprisingly, multi-junction cells do not only provide higher EH efficiencies, but also yield higher achievable rates and lower bit-error probabilities than single-junction cells, and hence, are preferable for SLIPT.

	\subsection{Rate-Power Tradeoff}
	\label{Section:Tradeoff}
	
\begin{figure*}[!t]
	\centering
	\includegraphics[draft=false, width=0.7\textwidth]{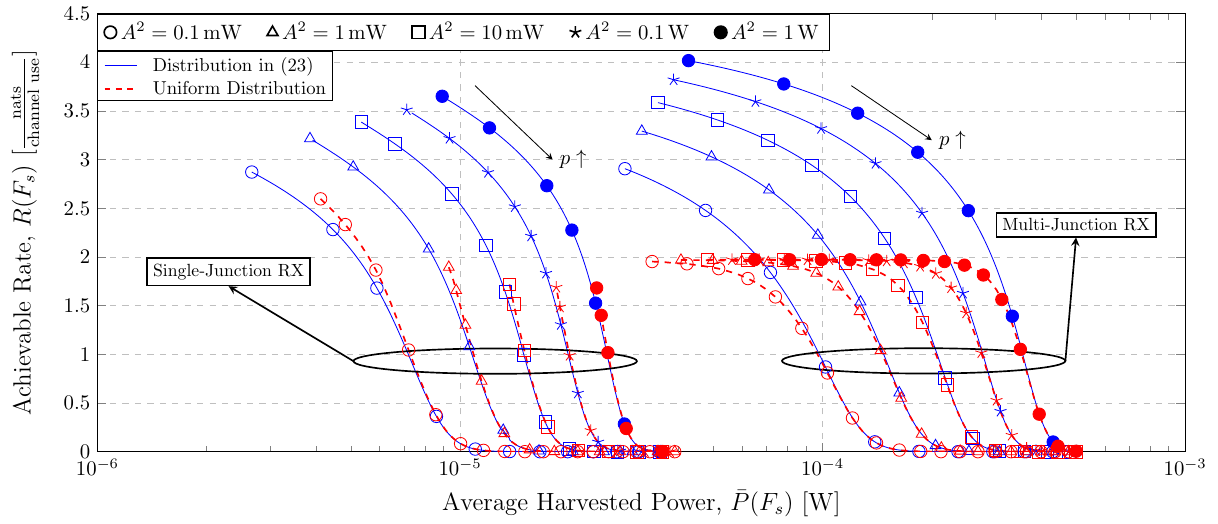}\vspace*{-7pt}
	\caption{Rate-power regions for different RX designs, TX signal pdfs $F_{s}$, and maximum FSO TX powers $A^2$.}
	\label{Fig:ResultsTradeoff}
	\vspace*{-10pt}
\end{figure*}

In the following, we investigate the tradeoff between the achievable information rate and the average harvested power for SLIPT systems with photovoltaic RXs.
To this end, for given maximum transmit information power $A^2$, by adjusting the powers $p\in\mathbb{R}_{+}$ of the energy FSO signals, we obtain and plot in Fig.~\ref{Fig:ResultsTradeoff} different values of the achievable information rate $R(F_{s})$ in (\ref{Eqn:RateBound}) and the average harvested power $\bar{P}(F_{s})$ for both single-junction and multi-junction RXs.
Since an increase in the ambient light intensity $\mu_\text{a}$ is equivalent to adopting a higher energy signal power $p$, the rate-power tradeoff is determined by $A^2$ and $N$, and thus, we set $\mu_\text{a} = 0$.
Similar to Fig.~\ref{Fig:ResultsAchievableRate}, for the results in Fig.~\ref{Fig:ResultsTradeoff}, we consider the proposed optimal cdf $F^*_{s}$, whose achievable rate and average harvested power are given by (\ref{Eqn:RateBound}) and (\ref{Eqn:ContAverageHarvestedPower}), respectively, and a baseline scheme, whose transmit information FSO signal $s$ is uniformly distributed.
For the baseline scheme, we calculate the achievable information rate with (\ref{Eqn:RateBoundGeneral}), whereas the average harvested power is given by $\bar{P}(F_{s}) = \int_0^1 P_\text{harv}(j^\text{s}; \boldsymbol{j}^\text{a}) \, \text{d} F_{s}(s)$. 

First, similar to the results in Figs.~\ref{Fig:ResultsValidationEH}, \ref{Fig:ResultsAchievableRate}, and \ref{Fig:ResultsBERs}, also in Fig.~\ref{Fig:ResultsTradeoff}, we observe that for both considered photovoltaic RX designs, higher information signal powers $A^2$ yield higher average harvested powers and higher achievable information rates.
Furthermore, we note that a larger number of p-n junctions $N$ at the RX leads to a higher average harvested power $\bar{P}(F_{s})$ and a higher achievable rate $R(F_{s})$.
Moreover, similar to Fig.~\ref{Fig:ResultsAchievableRate}, we observe also in Fig.~\ref{Fig:ResultsTradeoff} that the proposed cdf $F_{s}^*$ always shows significantly higher achievable information rates than the baseline scheme.
However, we note that since the proposed distribution in (\ref{Eqn:CapacityAchievingDistributionS}) maximizes the achievable information rate for given $A^2$ and $p$, it does not necessarily yield the maximum information rate for a given average harvested power at the optical RX.
Uniformly distributed transmit information signals may be able to provide higher achievable rates than the proposed distribution for a given $\bar{P}(\cdot)$, but this is achieved at the expense a significantly higher average information transmit signal power $s$, as shown in Fig.~\ref{Fig:ResultsCdfs}.
Furthermore, since higher energy signal powers $p$ yield higher average harvested powers and lower achievable rates, we observe that for SLIPT systems with non-linear photovoltaic RXs, there is a tradeoff between the data rate and harvested power that is characterized by the rate-power regions in Fig.~\ref{Fig:ResultsTradeoff}.
In particular, by adjusting the powers $p$ at the FSO energy TXs, it is possible to trade the information rate for a higher average harvested power at the optical RX.
Thus, we conclude that an accurate modelling of photovoltaic RXs is important not only for precise characterization of the harvested power, but also for reliable communication and efficient optical SLIPT system design.

\section{Conclusions}
	\label{Section:Conclusions}
	
In this work, we studied SLIPT systems, where the optical RX was equipped with a photovoltaic cell.
To enable efficient simultaneous EH and information reception, we proposed to employ multiple p-n junctions at the optical RX and considered the case where the RX is illuminated by ambient light, an information-carrying FSO signal, and an energy-providing optical signal.
The received optical signal at the photovoltaic RX was converted into an electrical signal, whose AC and DC components are separated and utilized for information decoding and EH, respectively.
We carefully modelled the current flow through the multi-junction photovoltaic RX and derived a novel accurate EH model, which characterizes the harvested power as a function of the received optical power.
Furthermore, since the accurate EH model did not lend itself to SLIPT system design, we also derived novel approximate closed-form analytical EH models for photovoltaic RXs with single and multiple p-n junctions, respectively.
Moreover, taking into account the non-linearities of photovoltaic RX circuits, we derived the distribution of the FSO transmit information signal that maximizes the achievable data rate, and for practical OOK modulated transmit signals, we determined the bit-error probability for maximum likelihood information detection at the RX.
We validated the proposed analytical EH models with circuit simulation results and demonstrated that multi-junction RXs not only have higher EH efficiencies, but also achieve higher data rates and lower bit-error probabilities compared to RXs equipped with a single p-n junction.
Furthermore, we showed that in contrast to the proposed EH model, two baseline models for a single-junction optical RX based on MPP tracking at the RX and a single-diode EEC are not able to fully capture the non-linear behavior of photovoltaic cells.
Moreover, we highlighted that the proposed transmit signal distribution yields significantly higher achievable rates than uniformly distributed information signals, which are optimal for linear information RXs.
Finally, we investigated the rate-power tradeoff for SLIPT systems and showed that by adjusting the energy-providing transmit signal power, the achievable data rate can be traded for a higher average harvested power at the optical RX.

\appendices
	\renewcommand{\thesection}{\Alph{section}}
	\renewcommand{\thesubsection}{\thesection.\arabic{subsection}}
	\renewcommand{\thesectiondis}[2]{\Alph{section}:}
	\renewcommand{\thesubsectiondis}{\thesection.\arabic{subsection}:}	
	\section{Proof of Proposition \ref{Theorem:Capacity}}
	\label{Appendix:ProofCapacity}
	The proof follows similar steps as that of \cite[Theorem 5]{Lapidoth2009}.
Exploiting the entropy power inequality, for a given $F_{s}$, we obtain the mutual information $I(y, s; F_{s})$ between the information TX and RX symbols $s$ and $y$ as follows \cite{Lapidoth2009}:
\begin{align}
	I(y , s ; F_{s} ) &=  u(y; F_{s}) - u(n) = u(x+n; F_{s}) - u(n) \nonumber\\
	& \geq \frac{1}{2} \ln \Big( 1+ \frac{e^{2u(x; F_{s})}}{ 2\pi e \sigma^2 } \Big) \triangleq R(F_{s}),
	\label{Eqn:RateBound2}
\end{align}
\noindent where $u(y; F_{s})$ and $u(n) = \frac{1}{2} \ln(2 \pi e \sigma^2)$ are the differential entropies of $y$ for given $F_{s}$ and AWGN, respectively.

We note that for a given FSO TX power $A^2$, the symbols at the RX output are bounded by $x[k] \in \mathcal{X} \triangleq [x_0, x_{A} ], \forall k$.
Therefore, the differential entropy $u(x)$, and hence, the achievable rate in (\ref{Eqn:RateBound2}) are maximized if the pdf $f^*_x$ of $x$ is uniform in $\text{dom} \{f^*_x\} = \mathcal{X}$ and is given by 
\begin{equation}
	f^*_x(x) = \frac{1}{ \theta(A^2) }.
	\label{Eqn:CapacityAchievingDistribution}
\end{equation}
Finally, the corresponding cdf of $s$ and the maximum achievable rate are given by (\ref{Eqn:CapacityAchievingDistributionS}) and (\ref{Eqn:RateBound}), respectively.
This concludes the proof.
	\section{Proof of Corollary \ref{Theorem:PowerCont}}
	\label{Appendix:ProofPowerCont}
	First, we note that for a given symbol $x$, the corresponding instantaneous harvested power can be obtained as $\rho(x) = x^2$, see (\ref{Eqn:CommChannel}).
Next, we recall that the capacity-achieving distribution of $x$ is uniform in $\mathcal{X}$, and for a given maximum information FSO power $A^2$, is given by (\ref{Eqn:CapacityAchievingDistribution}).
Thus, the corresponding pdf $f_\rho(\rho)$ of random variable $\rho$ is given by
\begin{align}
	f_\rho(\rho) &= 
	\begin{cases}
		\frac{1}{\theta(A^2)} \frac{1}{2 \sqrt{\rho} }, &\text{if} \; \rho \in [x_0^2, \, x_{A}^2 ], \\
		0, &\text{otherwise}.
	\end{cases} 
\end{align}
Finally, we obtain the average harvested power as follows:
\begin{align}
	\bar{P}_\text{\upshape harv}(A^2) &= \int_{\rho} \rho f_\rho(\rho) \text{d} \rho =
	\frac{1/3}{ \theta(A^2) }  \rho^{\frac{3}{2}} \bigg\lvert_{P_\text{harv} ( 0; \boldsymbol{j}^\text{a} )}^{P_\text{harv} ( j^\text{s}_\text{max}; \boldsymbol{j}^\text{a} )} \nonumber \\
	&= \frac{1}{3} ( x^2_{A} + x^2_0 + {x_{A}  x_0 } ),
\end{align}
\noindent where $j^\text{s}_\text{max} = h A^2 r_1(\lambda_0)$.
This concludes the proof.
	\section{Proof of Proposition \ref{Theorem:BER}}
	\label{Appendix:ProofBER}
		Since the transmit information symbols $s \in \{0, A^2\}$ are \gls*{iid}, due to the symmetry of the decision rule in (\ref{Eqn:MLDecoding}), the bit-error probability for maximum likelihood detection at the RX can be expressed as follows:
\begin{align}
	P_\text{\upshape e}(A^2) &=  \frac{1}{2}\Pr \{ y > \frac{x_0 + x_{A}}{2} \, \vert s = 0 \} \nonumber\\
	&\quad\quad\quad+ \frac{1}{2}\Pr \{ y < \frac{x_0 + x_{A}}{2} \, \vert s = A^2 \} \nonumber \\
	&=  \int_{\frac{ x_0 + x_{A} }{2}}^{+\infty} \frac{1}{\sqrt{2\pi}\sigma} \exp\big( - \frac{(x - x_0)^2}{2\sigma^2}\big) \text{d}x \nonumber \\
	&=  Q\bigg(\frac{\theta(A^2)}{ 2\sigma} \bigg).
	\label{Eqn:BER_Proof}
\end{align}
 This concludes the proof.

\bibliographystyle{IEEEtran}
\bibliography{WPT_Bibl.bib}

\end{document}